\documentclass[lettersize,journal]{IEEEtran}
\usepackage{amsmath,amsfonts}
\usepackage{array}
\usepackage[caption=false,font=normalsize,labelfont=normalfont,textfont=normalfont]{subfig}

\usepackage{caption}

\usepackage{textcomp}
\usepackage{stfloats}
\usepackage{url}
\usepackage{verbatim}
\usepackage{graphicx}
\usepackage{amssymb,setspace}
\usepackage{cite}
\usepackage{xcolor}
\usepackage{amsmath}
\usepackage{mathtools, nccmath}

\usepackage[ruled,vlined,linesnumbered]{algorithm2e}
\DeclareMathOperator*{\minimize}{minimize}
\newtheorem{theorem}{Theorem}

\newtheorem{corollary}{Corollary}

\newtheorem{definition}{Definition}

\newtheorem{notation}{Notation}

\newenvironment{proof}[1][Proof]{\noindent\textbf{#1.} }{\ \rule{0.5em}{0.5em}}
\usepackage{mathtools, nccmath}

\DeclareMathOperator{\Ima}{Image}
\begin{document}

\title{Iterative Planning for Multi-agent Systems: An Application in Energy-aware UAV-UGV Cooperative Task Site Assignments}

\author{Neelanga Thelasingha, Agung Julius, James Humann, Jean-Paul Reddinger, James Dotterweich,  Marshal Childers
        % <-this % stops a space

\thanks{This work was supported by DEVCOM Army Research
Laboratory under the Cooperative Research Agreement W911NF-21-2-0283 79054-CI-ARL}% <-this % stops a space
\thanks{Neelanga Thelasingha and Agung Julius are with the Center for Mobility with Vertical Lift (MOVE),
Rensselaer Polytechnic Institute, 2303 Jonsson Engineering Center, 110 8th Street, Troy, NY 12180, USA. (email: thelan@rpi.edu, agung@ecse.rpi.edu)}% <-this % stops a space
\thanks{ James Humann, Jean-Paul Reddinger, James Dotterweich and Marshal Childers are with DEVCOM Army Research Laboratory, Aberdeen Proving Ground, MD
21005, USA. (email: james.d.humann.civ@army.mil, jean-paul.f.reddinger.civ@army.mil, james.m.dotterweich.civ@army.mil, marshal.a.childers.civ@army.mil)}}

\markboth{Journal of \LaTeX\ Class Files,~Vol.~14, No.~8, August~2021}%
{Shell \MakeLowercase{\textit{et al.}}: A Sample Article Using IEEEtran.cls for IEEE Journals}


\maketitle

\begin{abstract} This paper presents an iterative planning framework for multi-agent systems with hybrid state spaces. The framework uses transition systems to mathematically represent planning tasks and employs multiple solvers to iteratively improve the plan until computational resources are exhausted. When integrating different solvers for iterative planning, we establish theoretical guarantees for recursive feasibility. The proposed framework enables continual improvement of solutions to reduce sub-optimality, efficiently using allocated computational resources. The proposed method is validated by applying it to an energy-aware UAV-UGV cooperative task site assignment problem. The results demonstrate continual solution improvement while preserving real-time implementation ability compared to algorithms proposed in the literature. 
\end{abstract}

\def\abstractname{Note to Practitioners}
\begin{abstract}
%Is the NtP addressed to industrial practitioners and sufficiently distinct from the Abstract (which is addressed to colleagues in research)?
%Does the NtP summarize the innovative aspects of the paper, the practical value of the results, and how the results might be applied in the near or mid-term from a practitioner's point of view?
%Does the NtP clearly describe practical limitations of the approach and directions for future research?
\color{black}
This paper presents an iterative planning solution for cooperative planning problems in multi-agent systems, which integrates multiple solvers to create an optimization framework. The proposed planning framework has been theoretically validated and applied in an energy-aware cooperative planning scenario for multi-vehicle task site assignments. The proposed framework can be applied to plan for any generalized task site assignment using multiple solvers iteratively. 
\color{black}

\end{abstract}

\def\abstractname{Submission Type}
\begin{abstract}
    Regular paper
\end{abstract}

\section{Introduction}
% \JDH{James's notes 8 DEC 2022. I couldn't use ``track changes'' so added some comments in red. I made a bunch of direct edits for typos and wording. I also tried to organize Section I into subsections. The Contributions subsection is kind of short so maybe could be merged into Overview. I only looked at Sections I and II and think they look good overall. One consistent thing I noticed is that ``throughout'' should always be one word, not ``through out.'' Happy to provide more info on any comments that are unclear.}
% \NT{Thank you, James! I have replied to your comments inline in blue and also edited the content accordingly. Please let me know if further edits are required.}
\IEEEPARstart{R}{ecent} advances in robotics and control have enabled multi-agent systems that operate in complex real-world environments. Operating such systems often involves solving complex planning problems in multi-dimensional state spaces. Real-world implementations of such solutions have wide applications, such as in Unmanned Aerial/Ground Vehicle (UAV/UGV) systems \cite{1_8671711,14_carlsson_song_2018}, underwater vehicles\cite{Yordanova2017RendezvousPF}, route planning for a multi-robot station in a manufacturing cell
% \JDH{what's a multi-robot station? \NT{Ensemble of cooperating robots used at a manufacturing plant} } 
\cite{3_9842330}, station inventory management systems \cite{4_8792385} and automated material handling systems \cite{5_8754729}. Many of these applications involve planning trajectories in high-dimensional state spaces. Among them, route planning tasks for UAV-UGV teams \cite{6_8798870,2_9476848,8_ramasamy_reddinger_dotterweich_childers_bhounsule_2022,10_ostertag_atanasov_rosing_2022,Redding2011DistributedMP} involve specific challenges due to the heterogeneous dynamics. Further, agents' operations are limited by several constraints, including energy usage \cite{Tatsch2020RoutePF,Mitchell2015MultiRobotPC} and the assigned task, for example, package delivery \cite{Mathew2015PlanningPF}.

\subsection{Motivation}
In general, most cooperative planning tasks can be formulated as a variant of the well-known NP-hard vehicle routing problem \cite{13_9926704,Yu2018AlgorithmsAE} \color{black} and its two-echelon vehicle routing problem counterpart \cite{https://doi.org/10.1111/itor.13052, ZHOU2023124}, but the complexity increases when energy dynamics must also be considered \cite{2_9476848,Maini2019CooperativeAV}.  \color{black} The planned routes must be feasible under the energy availability constraint for each vehicle.  When energy recharging is considered, additional constraints are imposed \cite{Kenzin2020CoordinatedRO}. In any solution, each agent must operate within the bounds of the energy capacity and follow the defined discharge/recharge dynamics\cite{Rappaport2017CoordinatedRO,Xing2021BatteryCD}. Therefore, including the energy dimension in energy-aware planning tasks further increases the complexity \cite{5_8754729,8_ramasamy_reddinger_dotterweich_childers_bhounsule_2022}. The state spaces in such problems are hybrid, with multiple continuous and discrete dimensions. When the applications include persistent or periodic delivery tasks that require receding horizon planning, such as the cases in \cite{3_9842330,Redding2011DistributedMP} and \cite{11_6042503}, the allowed computation resources per planning step are bounded.  
\color{black}

In the literature, many approaches include discretizing the state space to speed up computation. 
% Such approaches include graph-based state discretizations and hierarchical planning \cite{2_9476848,5_8754729,Mathew2013AGA}. 
These formulations enable the application of Traveling Salesperson Problem (TSP) solvers and Vehicle Routing Problem (VRP) solvers in the discretized state space \cite{8_ramasamy_reddinger_dotterweich_childers_bhounsule_2022,Yu2018AlgorithmsAE}
. Further, a discrete formulation of the problem allows the energy dynamics and task site assignment to be encoded as integer and logical constraints. Then, the cooperative route planning task can be solved as a Mixed-Integer Program (MIP)\cite{Sundar2015FormulationsAA}.
% Further, heuristic-based MIP solvers have been used, for example, in \cite{3_9842330,Maini2019CooperativeAV} and \cite{6_8798870}. Some  other papers, e.g.,\cite{2_9476848} and \cite{9_WEN2022103763} use hierarchical multi-level optimization and \cite{7_9836044} uses a genetic algorithm and Bayesian optimization.
However, solving the MIP for optimality is computationally expensive. Therefore, hierarchical multi-level optimization approaches have been proposed to reformulate the constraints or break the problem into smaller subproblems and solve them \cite{2_9476848,9_WEN2022103763}.  This still requires solving an NP-hard optimization problem, which scales up very quickly with the number of agents and task sites. 
Another approach is data-driven methods that leverage computation and training time to learn a policy that can be deployed \cite{10122127,dong2022review}. However, these methods suffer from sample inefficiency and lack of safety guarantees \cite{robotics11040081}.

We identify the need for a solver for the task site assignment problem with anytime computational property.
\color{black} This allows a solver to provide feasible but possibly suboptimal solutions at any point during its execution. \color{black} To this end, we propose a unifying iterative planning framework for multi-agent systems. Our framework first finds a feasible solution using a constraint solver and iteratively optimizes the solution, given the availability of computation resources. See the example of its implementation in Figure. \ref{fig:exsample_scn}, where a Satisfiability Modulo Theory (SMT) solver provides a feasible solution, and successive solvers improve upon the previous solutions. The challenge here is that these subsequent solvers operate in varying hybrid state spaces. Therefore, a unifying framework where consensus can be achieved by all solvers is essential. In this paper, we proposed to use the concept of transition systems as the basis for this unifying framework \cite{Belta_Yordanov_Gol,tabuada_book}.

\begin{figure}[h]
\centering
\includegraphics[width=0.5\textwidth]{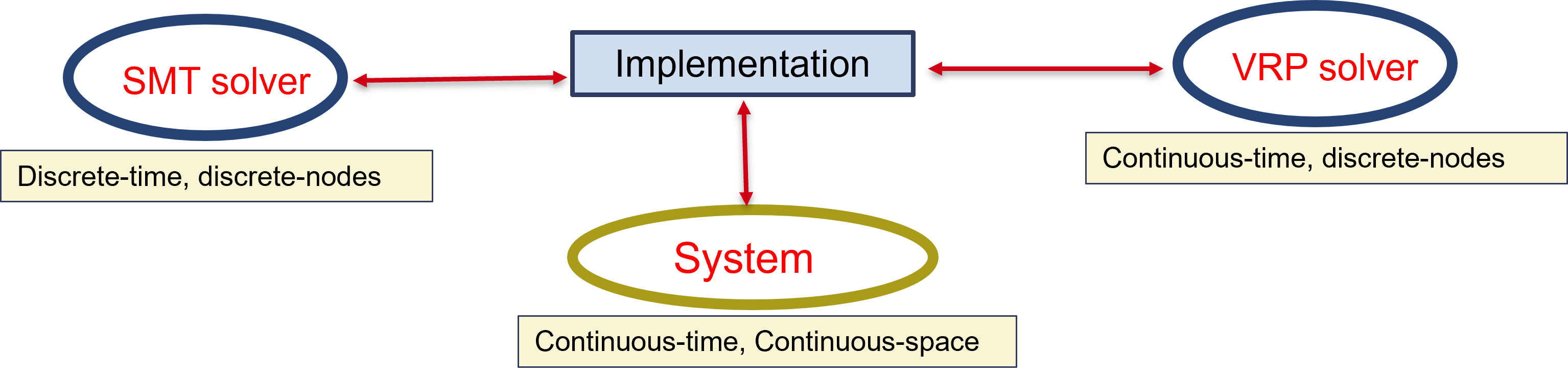}
\caption{\color{black} Example - SMT solver that operates in discrete-time on discrete-nodes and a VRP solver that operates in continuous-time on discrete-nodes generating implementations for a system in continuous-time on continuous-space.}
\label{fig:exsample_scn}
\end{figure}

\color{black}

\subsection{Overview}
Transition systems have been widely used as a theoretical formulation to represent problems in complex systems with state spaces with continuous and discrete dimensions\cite{Belta_Yordanov_Gol}. Standard problem discretization procedures can be applied to formulate sampled transition systems and enable faster solution computation. Our proposed framework uses multiple solvers that are called iteratively to improve the existing solution. An example approach is to call a faster solver initially to generate a feasible plan and then to iteratively call the next solvers to improve the plan until computational resources are exhausted. However, the solvers have different sampling and discretization schemes and therefore require a unifying framework to support integration. Therefore, we derive and prove the necessary mathematical conditions to guarantee the recursive feasibility, i.e., that all solvers agree on each other's solutions and identify them as feasible plans to satisfy the task site assignment.

% To this end, we propose to maintain a general data-structure for the specification and the current plan such that it is a feasible solution under each solver. Further, we propose a solver wrapper as the interface between any solver and the current plan. This ensures that all solvers in the framework agree on the feasibility and the specification satisfaction of a plan proposed any other solver.

%Further, in the proposed solution, any solver combination or iteration arrangement can be made.  This ensures that a feasible solution is always computed and stored while the latter solvers conduct the improvement. The proposed framework imposes constraints on plan feasibility and requires improvement from each solver called iteratively. This ensures solution improvement at each planning step.  Moreover, hierarchical optimization at the team-level and agent level can be integrated with the proposed unifying framework.

% In this paper, we propose a framework for iterative planning in a multi-agent system. We use transition systems to mathematically define the planning problem. The proposed framework ensures that a feasible solution is computed and always kept at hand. Meanwhile, the solution is further optimized by subsequent solvers.
We apply the proposed framework on an energy-aware UAV-UGV cooperative routing task. In this application, we formulate team-level and agent-level planning problems. The effectiveness of our framework is demonstrated through improving the objective function, i.e., the plan execution time, when compared to standard optimization solvers like mixed integer programming solvers \cite{6_8798870, Maini2019CooperativeAV} used in state-of-the-art solutions. 

\subsection{Contributions}
\color{black}
The main contributions of this paper are listed and described as follows.
\begin{enumerate}
\item Iterative planning framework that integrates multiple solvers for planning in multi-agent systems.
The iterative planning framework enables fast computation of a feasible plan and continual improvement of the plan given computation time. The framework also includes a shrinking horizon implementation that replans at each step of execution for further optimization over time. 

\item Theoretical analysis of recursive feasibility of the iterative planning framework.
Using transition systems, we introduce a mathematical framework to unify different solvers for planning in multi-agent systems with hybrid state spaces for any general task site assignment problems. We derive and prove the conditions to preserve the recursive feasibility of the proposed algorithms. 

\item Application in energy-aware UAV-UGV cooperative task site assignments.
We apply the proposed framework to solve energy-aware UAV-UGV cooperative task site assignments. We cast the problem in the proposed mathematical framework to prove the recursive feasibility for two solvers that are called iteratively to optimize the plan. The shrinking horizon execution enables the continual improvement of the plan by re-planning at each step.

\item Comparative result analysis of proposed approach against optimization solvers and state-of-the-art methods.
We compare the compute time and total plan time of the proposed approach against standard MIP solver and state-of-the-art methods proposed in literature such as Bayesian optimization and evolutionary methods. We also include a sensitivity analysis of the proposed method against the complexity of the problem by varying the number of task sites and the number of agents. 
\end{enumerate}

\section{Related Work}

In the literature, there are solutions for multi-agent planning and routing for specific tasks, such as package delivery and distribution assignments \cite{1_8671711,14_carlsson_song_2018, Zhen2022BranchpriceandcutFT}. However, adapting these solutions to the generalized tasks site assignment, which includes specifications on required visits and some general tasks at the intended sites is challenging \cite{1_8671711,2_9476848,8_ramasamy_reddinger_dotterweich_childers_bhounsule_2022,10_ostertag_atanasov_rosing_2022}. Some literature considers agents that do not interact \cite{1_8671711,9_WEN2022103763}. \color{black} Others consider agent-agent interactions, for example, between ground vehicles and aerial vehicles in package delivery scenarios \cite{https://doi.org/10.1002/net.22019,14_carlsson_song_2018,ZHOU2023124,https://doi.org/10.1111/itor.13052} or as a refueling or recharging resource in surveillance-type tasks \cite{8_ramasamy_reddinger_dotterweich_childers_bhounsule_2022,6_8798870,2_9476848,Redding2011DistributedMP,Mathew2015PlanningPF} where cooperative tasks can be formulated as a two-echelon vehicle routing problem. \color{black} Our paper captures agent-agent interaction in a generalized task site assignment.

Energy-aware planning problems \cite{4_8792385,5_8754729,2_9476848,8_ramasamy_reddinger_dotterweich_childers_bhounsule_2022,13_9926704,Yu2018AlgorithmsAE,Kenzin2020CoordinatedRO,7_9836044} arise when energy consideration is included in planning, adding further complexities by modeling the energy dynamics \cite{Sundar2015FormulationsAA} and optimal planning for the refueling/recharging instances\cite{Xing2021BatteryCD}.   In this paper, we formulate the planning problem for cooperative tasks as an energy-aware planning problem. Usual energy-aware planning problem formulations include the VRP, which is a generalization of the well-known NP-hard traveling salesperson problem \cite{10_ostertag_atanasov_rosing_2022,Tatsch2020RoutePF,Mathew2015PlanningPF,Cone,Marler:2023:0028-1425:115}. Adding more constraints such as health monitoring \cite{Redding2011DistributedMP}, time constraints \cite{Tatsch2020RoutePF,11_6042503, 7_9836044}, and possible adversarial attacks \cite{Cone, Marler:2023:0028-1425:115} further exacerbates the problem complexity. In this work, we first focus on finding feasible solutions and iteratively optimize them later.

Once the planning problem is formulated, the solution method depends on the problem variable representation.  Solving the problem optimally in the naturally continuous state space would yield optimal solutions\cite{Yordanova2017RendezvousPF}, but with high computation complexity. On the other hand, a discretization \cite{6_8798870} of the state space with graphical representations \cite{Mathew2015PlanningPF,Mathew2013AGA} enables the use of VRP solvers and MIP solvers \cite{8_ramasamy_reddinger_dotterweich_childers_bhounsule_2022,2_9476848}. In recent work, optimal solutions have been achieved by solving an MIP \cite{6_8798870, Maini2019CooperativeAV}. However, this is not applicable for large problems \cite{Mathew2013AGA, Sundar2015FormulationsAA,Cone}. Some papers, for example, \cite{6_8798870, Tatsch2020RoutePF, Mitchell2015MultiRobotPC, Mathew2013AGA}, use a heuristic-based approach that is challenging to generalize. There are some other techniques such as adaptive neighborhood search \cite{https://doi.org/10.1002/net.22019} that can improve a solution. However, obtaining the optimal solutions for a large number of heterogeneous interacting agents is not possible.  

Another class of recent work proposes to use a probabilistic formulation of the planning problem \cite{Redding2011DistributedMP,1_8671711}. Solving such problems usually requires a data-driven learning-based approach\cite{Rappaport2017CoordinatedRO,9712866}, which learns an optimal policy \cite{Yordanova2017RendezvousPF}. However, this approach usually requires a lot of samples and extensive training. Further, it is not possible to establish safety or performance guarantees with such approaches, unlike greedy methods \cite{10_ostertag_atanasov_rosing_2022}.

Some published work uses a hierarchical approach, where an initial step of clustering of task sites simplifies the problem for the next step\cite{8_ramasamy_reddinger_dotterweich_childers_bhounsule_2022,9001184,2_9476848,Maini2019CooperativeAV,7_9836044}. In the next step,   \cite{5_8754729,6_8798870,Sundar2015FormulationsAA,Cone} use a MIP solver while \cite{2_9476848,8_ramasamy_reddinger_dotterweich_childers_bhounsule_2022,7_9836044} use Bayesian optimization. However, an optimization problem still needs to be solved to obtain a specification-satisfying solution.

In contrast with the previous approaches, our proposed approach uses an iterative framework, where, at the first step, a constraint solver generates a feasible solution. Then, given the availability of the computation resources, it is further optimized through iterative solvers. Here, it is important to note that at each step of the iterative framework, a feasible solution always exists at hand. This formulation requires a multi-level optimization framework \cite{9001184} for multiple solvers with \color{black}hybrid state spaces.

We use transition systems, commonly used in literature for modeling complex problems with hybrid state spaces \cite{doi:10.1146/annurev-control-091420-084139}, to define the framework and conduct verification. In \cite{Abdulla_2010}, monotonic transition systems are used for symbolic reachability analysis, while \cite{5716670} and \cite{6362199} conduct verification in fuzzy automata. Further, \cite{4982634} and \cite{Girard_Julius_Pappas_2007} present simulation functions for stochastic hybrid automata using transition systems. Their modeling capabilities have been applied for decision problems such as signaling in traffic networks in \cite{10.1145/2728606.2728607}. Further, in \cite{10.5555/3091125.3091228} epistemic transition systems are defined to represent coalitions and distributed knowledge. Therefore, transition systems present a well-defined framework with modeling capabilities to generate verifiable results in hybrid systems.
\color {black}

\section{Preliminaries}

In this section, we introduce and elaborate on the required concepts and definitions used in this paper. 

\subsection{Notation}
\noindent$S$: Infinite state transition system. \\
$X$: Set of states in the transition system $S$. \\
$T$: Set of transitions in the transition system $S$. \\
$L$: Set of labels associated with transitions in $T$. \\
$Y$: Set of task sites in the transition system $S$. \\
$H$: Output map of the transition system $S$. \\
$p^j_{x}$: x-coordinate of  the $j^{th}$ vehicle. \\
$p^j_{y}$: y-coordinate of the  $j^{th}$ vehicle. \\
$e_j$: Energy level of the $j^{th}$ vehicle. \\
$f_j$: Label of the $j^{th}$ vehicle. \\
$x_j$: State of the $j^{th}$ vehicle, with components $[p^j_{x}, p^j_{y}, e_j, f_j]$. \\
$x$: System state. \\
$\gamma_d$: Sampling interval for the time dimension. \\
$\gamma$: Timestamp for the implementation of the trajectory. \\
$X_d$: Sampled state space. \\
$Y_d$: Sampled output space. \\
$T_d$: Set of transitions in the sampled transition system $S_d$. \\
$\tau$: Execution trajectory in the original transition system $S$. \\
$M$: Specification in the original transition system $S$. \\
$M_d$: Discretized specification. \\
$\Xi$: Task site assignment. \\
$\Sigma(S)$: Set of all execution trajectories in $S$. \\
$\Sigma^{N}(S)$: Set of all execution trajectories of length $N$ in $S$. \\
$\zeta_\tau$: Implementation of a trajectory $\tau$. \\
$\xi_\zeta(\gamma)$: Output at timestamp $\gamma$ of the implementation $\zeta$. \\
$N_y$: Cardinality of the assignment for the task site. \\
$\hat{Y_i}$: Output class in the task site assignment. \\
$\hat{X_i}$: Key state class formed by output class $\hat{Y_i}$. \\
$\hat{T_j}$: Key transition class formed by output class $\hat{Y_i}$. \\
$f_S$: Objective function in transition system $S$. \\
$t^{\max}_{\mathrm{compute}}$: Maximum allowed compute time for the solver. \\
$t^{\Omega}_{\mathrm{compute}}$: Compute time spent by the solver $\Omega$.\\
$\Ima(\xi_\zeta)$: Image of the output behavior of $\zeta$. \\
$\tau_d$: Trajectory in the sampled transition system . \\
$\zeta_\Sigma(S)$: Set of all feasible implementations in the system $S$.\\
\color{black}
\subsection{Transition Systems}
\color{black}
Many real-world planning tasks involve complex dynamical systems with finite and infinite state spaces along with their combinations. 
A proper mathematical formulation of these tasks and involved systems is essential when formulating planning strategies. Systems that are modeled as transition systems offer a well-defined framework for describing planning tasks and strategies. These systems provide a formal framework that captures the discrete changes and interactions within a system's state space over time. By representing state transitions and the corresponding actions or events, transition systems offer a concise and abstract representation of the underlying dynamics.

The advantages of using transition systems lie in their ability to handle both finite and infinite state spaces, accommodating a wide range of real-world scenarios. They enable the formulation of planning problems by defining the set of possible states, actions, and transitions between states. This clarity in representation facilitates the development of planning strategies, allowing systematic exploration of various solvers or solving strategies.  However, alternative formalisms such as differential equations or probabilistic models might be preferred in cases where continuous dynamics play a crucial role or uncertainty is inherent.

Throughout this paper, we use the idea of transition systems as the mathematical model that describes a dynamic phenomenon\cite{tabuada_book}. This enables us to model systems that evolve in finite state spaces with discrete transitions as well as systems with infinite and continuous state spaces that are described by differential equations. We formally define a transition system as follows.

\color{black}

\begin{definition}
[Transition System]
% \JDH{How is this different from a Markov Process? \NT{A transition $t$ can be deterministic as described by differential equations or stochastic as described by transition probabilities. Thus Markov processes are a subset of transition systems.}}
A transition system is defined as a quintuple $S=(X,T,L,Y,H),$
where $X$ is the state space (or set of states), $L$ is the set of transition labels, and $T\subset X\times L\times X$ is the set of transitions, where a transition

\[
t=(x,\ell,x^{\prime})
\]
denotes a transition from a state $x\in X$ to the state $x^{\prime}\in X,$
which is labeled by the symbol $\ell \in L$. Here, $\ell$ denotes the \underline{time duration} for the transition $t$ as shown in Figure \ref{fig:trns}. The set of outputs is denoted by $Y$ and $H\colon X \to Y$ is the output map. 

\end{definition}

\begin{figure}[h!]
\centering
\includegraphics[width=0.35\textwidth]{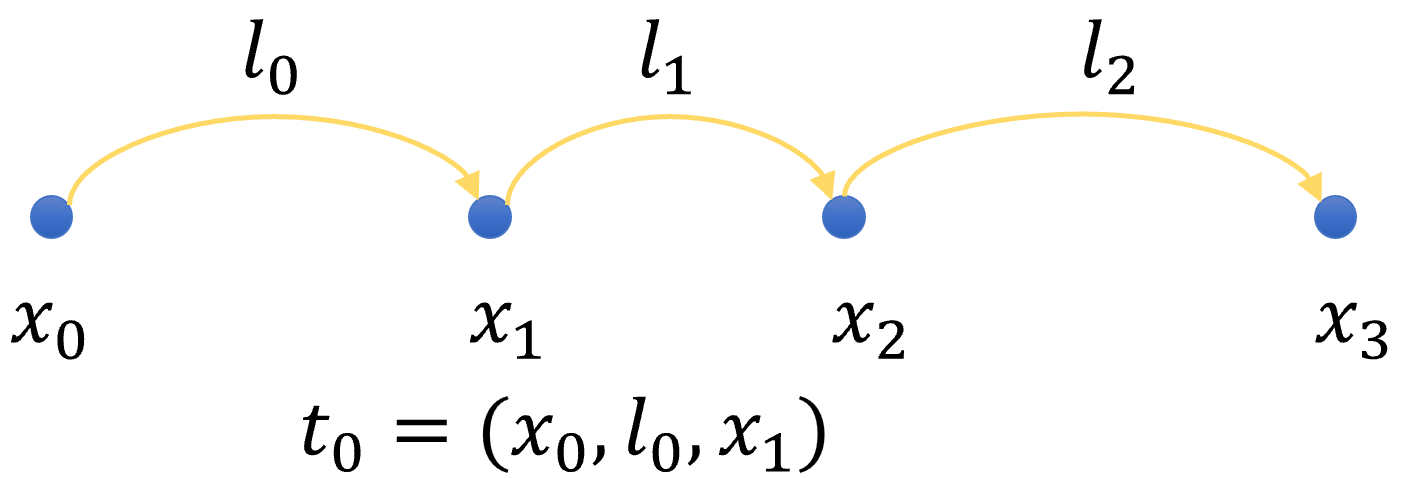}
\caption{\color{black} States, labels and transitions.}
\label{fig:trns}
\end{figure}

\begin{notation}
For a transition system $S=(X,T,L,Y,H)$, the fact that $(x,\ell,x^{\prime})\in T$
is also denoted by $x\overset{\ell}{\rightarrow_{T}}x^{\prime}$ or
$x\overset{\ell}{\rightarrow}x^{\prime}$ if the set of transitions $T$ is
clear from the context.
\end{notation}

\begin{notation}
The set of all strings (i.e., finite sequences) of symbols in $L$ is denoted
by $L^{\ast}$.
\end{notation}

\subsection{Trajectories and Implementations}

We define the evolution of a dynamic phenomenon within a transition system through sequential transitions as an execution trajectory.

\begin{definition}
% \JDH{This definition is confusing to me. Are 1) and 2) conditions for the definition? I think the iff should not have a semicolon after it}
% \NT{Edited the definition}
[Execution Trajectories] For any positive integer $N$, a sequence $\tau
:=x_{0}\ell_{0}x_{1}\ell_{1}\cdots\ell_{N-1}x_{N}$ where $x_{0},\cdots
,x_{N}\in X$ and $\ell_{0},\cdots,\ell_{N-1}\in L$ is said to be an execution
trajectory of $S=(X,T,L,Y,H)$ iff transition $t_k=(x_{k},\ell_{k},x_{k+1})$ exists in $S$, i.e., $t_k \in T$ for all $ k\in\{0,\cdots,N-1\}$. 
% where each $\ell_k \in \{\ell_{0},\cdots,\ell_{N-1}\}$ denotes the \underline{time duration} for the transition $t_k$. 
The length of the execution trajectory $\tau$ is defined as $N$.
\end{definition}

\begin{notation}
For a transition system $S=(X,T,L,Y,H),$ the set of all execution trajectories of
$S$ is denoted by $\Sigma(S)$. The set of all execution trajectories of length
$N$ of $S$ is denoted by $\Sigma^{N}(S)$.
\end{notation}

% \begin{definition}
% [Paths]Given a transition system $S$ and its execution trajectory $\tau
% =x_{0}\ell_{0}x_{1}\ell_{1}\cdots\ell_{N-1}x_{N}$, the path corresponding to
% this execution trajectory is the sequence of states $\xi=x_{0},x_{1}%
% ,\cdots,x_{N}$. The set of all paths of $S$ is denoted by $\Xi(S)$. The set of
% all paths of length $N$ of $S$ is denoted by $\Xi^{N}(S)$.
% We also use the notation $\xi_\tau$ to denote a path corresponding to a trajectory $\tau$
% \end{definition}

\begin{notation}
For a transition system $S=(X,T,L,Y,H),$ any state $\hat{x}\in X$ is said to be in an execution trajectory $\tau
:=x_{0}\ell_{0}x_{1}\ell_{1}\cdots\ell_{N-1}x_{N}$ if $\hat{x}=x_{i}$ for some $i \in \{0,\cdots,N\}$ and it is denoted by $\hat{x}\in \tau$. 
\end{notation}

\begin{notation}
For a transition system $S=(X,T,L,Y,H),$ any transition $t \in T$ is said to be in an execution trajectory $\tau
:=x_{0}\ell_{0}x_{1}\ell_{1}\cdots\ell_{N-1}x_{N}$, if $t=(x_{i},\ell_{i},x_{i+1})$ for some $i \in \{0,\cdots,N-1\}$ and it is denoted by $ t \in \tau$. 
\end{notation}

Next, we define implementations of trajectories. An implementation describes the system state's behavior when a trajectory is executed in real-time, following the transitions. Note that we associate the labels of transitions with the time duration of that transition execution.

\begin{definition}
[Implementation] Let a transition system be $S$ and an execution trajectory be $\tau
=x_{0}\ell_{0}x_{1}\ell_{1}\cdots\ell_{N-1}x_{N}$. Then an implementation of the trajectory $\tau$ is defined as a function of time $\zeta_\tau \colon (0,\sum_{i=0}^{N-1} \ell_i) \to X$ such that, for all $i\in \{0,1,\dotsc,N\}$, $\zeta_\tau(\sum_{k=0}^{i-1} \ell_k)=x_i$. We denote the set of implementations of all trajectories of $S$ as $\zeta_{\Sigma(S)}$.

\end{definition}

\subsection{Monotone Transition Systems}

Here, we introduce an ordering relation $\succeq$ on the states of a transition system. This allows comparison and ordering of the states according to some goodness measure. For example, let the state $x \triangleq  [x_p,x_f]$ where $x_p\in \mathbb{R}^2$ denotes the position of an agent and $x_f \in \mathbb{R}$ represents the fuel level. Let us define the fuel level as the goodness measure. Thus, given a state $x = [x_p,x_f]\in X$, any state $x^\prime = [x_p,x^{\prime}_f]\in X$ is greater or equal in the partial ordering if $x^{\prime}_f \geq x_f$ and it is denoted as $x^\prime \succeq x$.  We
assume that $\succeq$ is \emph{reflexive}, i.e., $x\succeq x$ for any $x\in X.$ Next, we define the monotone properties of transition systems in terms of partial ordering.

% \begin{notation}
% For a transition system $S=(X,T,L,Y,H),$ and execution trajectories $\tau :=x_{0}\ell_{0}x_{1}\ell_{1}\cdots\ell_{N-1}x_{N}$ and $\tau^\prime :=x_{0}^\prime\ell_{0}x_{1}^\prime\ell_{1}\cdots\ell_{N-1}x_{N}^\prime$, if $x_i^\prime \succeq x_i$ for some $i \in \{0,\cdots,N-1\}$, we lift the partial ordering  $\succeq$ to the set of trajectories and it is denoted by $ \tau^\prime  \succeq_\tau \tau$. 
% \end{notation}

% \JDH{Can you explain this in simpler terms somewhere? \NT{I added a sentence after the implication to explain this and also added a notation to denote the improved path.}
 % Is the ``implication'' saying that if you start in a better state and use all the same transitions, you will be in a better state at every step of the trajectory? \NT{Yes, for example, if you start with a state with better energy, given the same sequence of transitions you will be at better energy at each next step of the trajectory.} }
 
\begin{definition}
[Monotone Transition Systems] \label{def:monotonesys}Given a transition system  $S=(X,T,L,Y,H)$ and a
partial order relation $\succeq$ on $X$, $S$ is said to be monotone with
respect to $\succeq$ iff the following statement is true for every
$x_{1},x_{1}^{\prime}\in X$ and $\ell\in L$, such that $(x_{1},\ell,x_{1}^{\prime}) \in T$: For all $x_{2}\in X$ such that $x_{2}\succeq
x_{1},$ there exists a transition $(x_{2},\ell,x_{1}^{\prime})\in T$.
\end{definition}
% \color{black}
% \begin{definition}
% [Monotone Transition Systems]Given a transition system  $S=(X,T,L,Y,H)$ and a
% partial order relation $\succeq$ on $X$, $S$ is said to be monotone with
% respect to $\succeq$ iff the following statement is true for every
% $x_{1},x_{1}^{\prime}\in X$ and $\ell^{\prime}\in L$, such that $x_{1}\overset{\ell^{\prime}
% }{\rightarrow}x_{1}^{\prime}$: For all $x_{2}\in X$ such that $x_{2}\succeq
% x_{1},$ there exists an $x \in X$ and $\ell \in L$  such that $x_{2}^{\prime} \succeq x \succeq 
% x_{1}^{\prime}$ and $x_{2}\overset{\ell}{\rightarrow}x$ where  $x_{2}\overset{\ell^{\prime}}{\rightarrow}x_{2}^{\prime}$. We
% assume that $\succeq$ is \emph{reflexive}, i.e., $x\succeq x$ for any $x\in
% X.$
% \end{definition}

\color{black}
An implication of a transition system  $S=(X,T,L,Y,H),$ being monotone with respect to $\succeq$ is that if a sequence $x_{0}\ell_{0}x_{1}\ell_{1}\cdots\ell
_{N-1}x_{N}$ is known to be an execution trajectory of $S$, then for any
initial state $\hat{x}_{0}$ such that $\hat{x}_{0}\succeq x_{0},$ there is an
execution trajectory $\hat{x}_{0}\ell_{0}\hat{x}_{1}\ell_{1}\cdots\ell
_{N-1}\hat{x}_{N}$ where $\hat{x}_{1}\succeq x_{1},\cdots,\hat{x}_{N}\succeq
x_{N}.$ 

\color{black}
Therefore, in a monotone transition system, for a trajectory $\tau=x_{0}\ell_{0}x_{1}\ell_{1}\cdots\ell
_{N-1}x_{N}$, replacement of a state $x_i\in \tau$ by a state $\hat{x}_i$ such that $\hat{x}_i \succeq x_i$, results in a trajectory $\hat{\tau}$ with $\hat{x}_j \succeq x_j$ for all $j$ such that $j>i$.

\begin{notation}
For a transition system $S=(X,T,L,Y,H),$ that is monotone with respect to  $\succeq$, if $\tau=x_{0}\ell_{0}x_{1}\ell_{1}\cdots\ell
_{N-1}x_{N}$ is known to be a trajectory of $S$, then for any initial state $\hat{x}_{0}\succeq x_{0},$ there is an execution trajectory, $\hat{\tau}=\hat{x}_{0}\ell_{0}\hat{x}_{1}\ell_{1}\cdots\ell
_{N-1}\hat{x}_{N}$ where $\hat{x}_{1}\succeq x_{1},\cdots,\hat{x}_{N}\succeq
x_{N}$, which is denoted by $\hat{\tau}\succeq_{\tau} \tau$. Here, $\succeq_{\tau}$ denotes the partial ordering notation lifted to the set of trajectories.
\end{notation}

\color{black}

\subsection{Output Behavior}

While an execution trajectory exists in the state space of a system, observations can be made in the output space. Similarly, for any implementation that returns the system state at any given time, observations can be made in the output space by passing the returned state through the output map. The corresponding output sequence of an implementation observed through the output map $H$ is defined as the Output Behavior of an implementation.
% \begin{definition}
% [Output Behavior]For any execution trajectory $\tau=x_{0}\ell_{0}x_{1}\ell_{1}\cdots\ell_{N-1}x_{N}$ of a transition system $S=(X,T,L,Y,H)$,  sequence $\xi_\tau
% :=y_{0}y_{1}\cdots y_{N}$ where,
% \begin{equation}
% y_i=H(x_i),~\forall i\in\{0,\cdots,N\},
% \end{equation} is defined as the external behavior of a trajectory $\tau$.
% \end{definition}

% For any implementation that returns the system state at any given time, observations can be made in the output space by passing the returned state through the output map. To this end, We define the external implementation.

\begin{definition}
[Output Behavior]For a given execution trajectory $\tau=x_{0}\ell_{0}x_{1}\ell_{1}\cdots\ell_{N-1}x_{N}$ of a transition system $S=(X,T,L,Y,H)$, and its implementation $\zeta$,  a function of time $\xi_\zeta \colon (0,\sum_{i=0}^{N-1} \ell_i) \to Y$ where,
\begin{equation}
\xi_\zeta(\gamma)=H(\zeta(\gamma)),~\forall \gamma \in\big(0,\sum_{i=0}^{N-1} \ell_i\big),
\end{equation}
is defined as the output behavior of an implementation $\zeta$.
\end{definition}

\color{black}
\subsection{Task Site Assignment in a Transition System}

The goal of a task site assignment is to observe a specific set of outputs in the output space of a transition system.

\begin{definition}
[Task Site Assignment] For a transition system $S=(X,T,L,Y,H)$, a task site assignment is a collection of output classes, $\Xi=\{\hat{Y_1},\hat{Y_2},\dotsc,\hat{Y_{N_y}}\}$, where each $\hat{Y_i}\subset Y$. Here $N_y$ is the cardinality of the assignment. If there exists an implementation $\zeta$ and its output behavior $\xi_\zeta$ such that $\hat{Y_i} \cap \Ima(\xi_\zeta) \neq \emptyset $ for all $i=1,2,\dotsc, N_y$, then $\xi_\zeta$ is said to satisfy the task site assignment $\Xi_N$.
\end{definition}

 A planning problem is formulated as finding an execution trajectory in the transition system that produces the given outputs of the task site assignment when implemented.

\subsection{Sampled Transition System}

We propose a solution framework that can be employed to generate feasible yet sub-optimal solutions faster and optimize them given the computation resources and time. To this end, we discretize the infinite state system in focus through a sampling procedure as shown in Figure \ref{fig:sampled}. This involves a redefinition of the system as a sampled transition system. 
\begin{figure}[h!]
\centering
\includegraphics[width=0.4\textwidth]{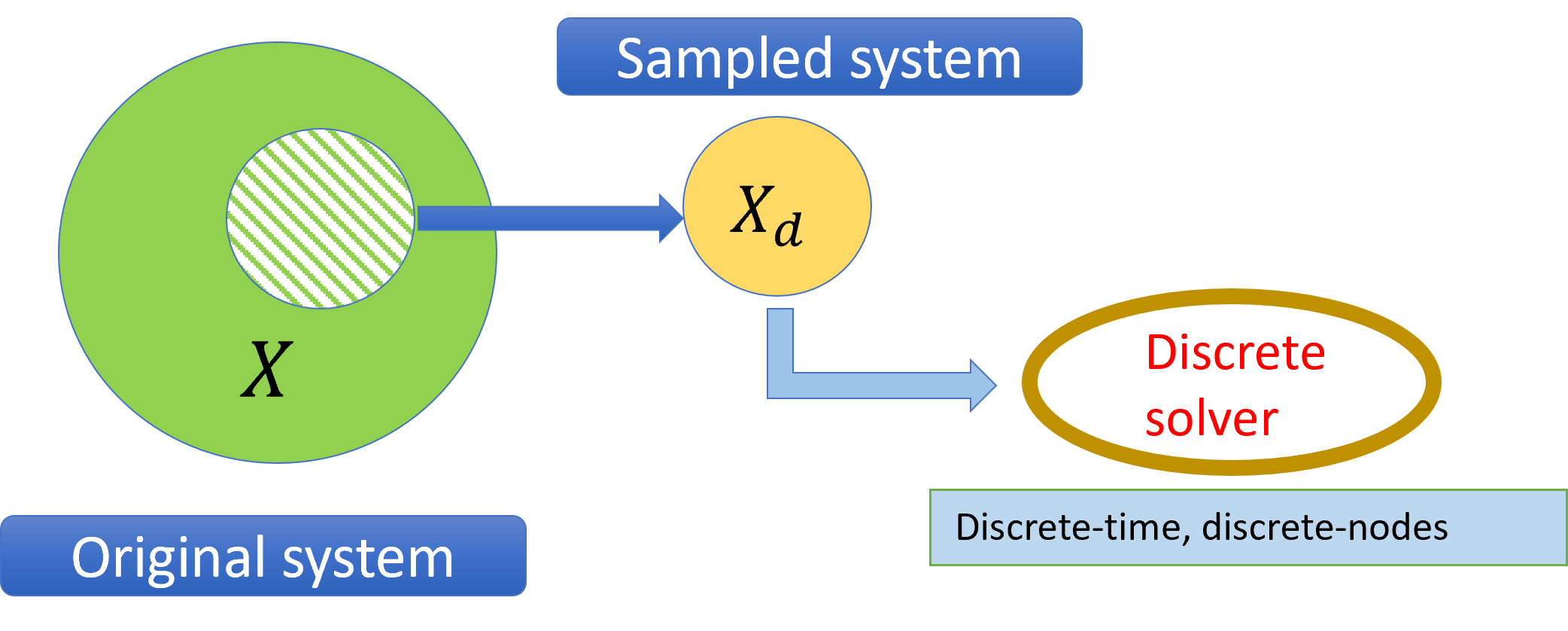}
\caption{\color{black} Sampled transition system for a discrete solver.}
\label{fig:sampled}
\end{figure}

\begin{definition}[Sampled Transition System]\label{def:sampled}
    Given an infinite state transition system, $S=(X,T,L,Y,H)$, we define a sampled transition system $S_d=(X_d,T_d,L,Y_d,H)$ by sampling each dimension of the state such that the sampled state space $X_d \subset X $. The set of transitions is denoted as $T_d = (X_d \times L \times X_d) \cap T$. Further, the set of outputs $Y_d=\{y\in Y\ |\  \exists x_d \in X_d \ \ s.t. \ \ y=H(x_d)\}$. 
\end{definition}

The sampled transition system $S_d=(X_d,T_d,L,Y_d,H)$ is now fully defined. Next, we define the specifications in the sampled transition systems, their key states, and key transitions.

% \begin{definition}[Discrete Mission]
%         Let $M$ be a mission in an infinite state transition system $S=(X,T,L,Y,H)$. Given a sampled transition system $S_d=(X_d,T_d,L,Y_d,H)$, equivalent discrete mission to $M$ is defined as $M_d=M \cap \Xi(S_d)$ such that the following conditions are satisfied:\newline(1) If for each key state class $j$ of $M$, such that $\hat{x}^j_{i}\in X_d$ for some $\hat{x}^j_{i}\in \hat{x}^j$, then $j$ is a key state class of $M_d$.
%         \newline(2) For each key state class $j$ of $M$, such that $\hat{x}^j_{i}\notin X_d$ for any $\hat{x}^j_{i}\in \hat{x}^j$, the key state class $j$ forms a key transition $\hat{t}^j\in T_d$ of $S_d$. 
% \end{definition}

% Thus, any path $\xi_d \in M_d $ satisfies the original mission $M$ in the infinite state system. 

\subsection{Key States, Key Transitions,  and Specifications}

\begin{figure*}[t]
  \centering
  \subfloat[Transition system $S$ with state space $X=\{A,B,C,D\}$. Key states are  $\{A,B,C,D\}$ where trajectory $\tau=A,\ell_0,B,\ell_1,C,\ell_2,D$  contains them.]{\includegraphics[width=0.4\textwidth]{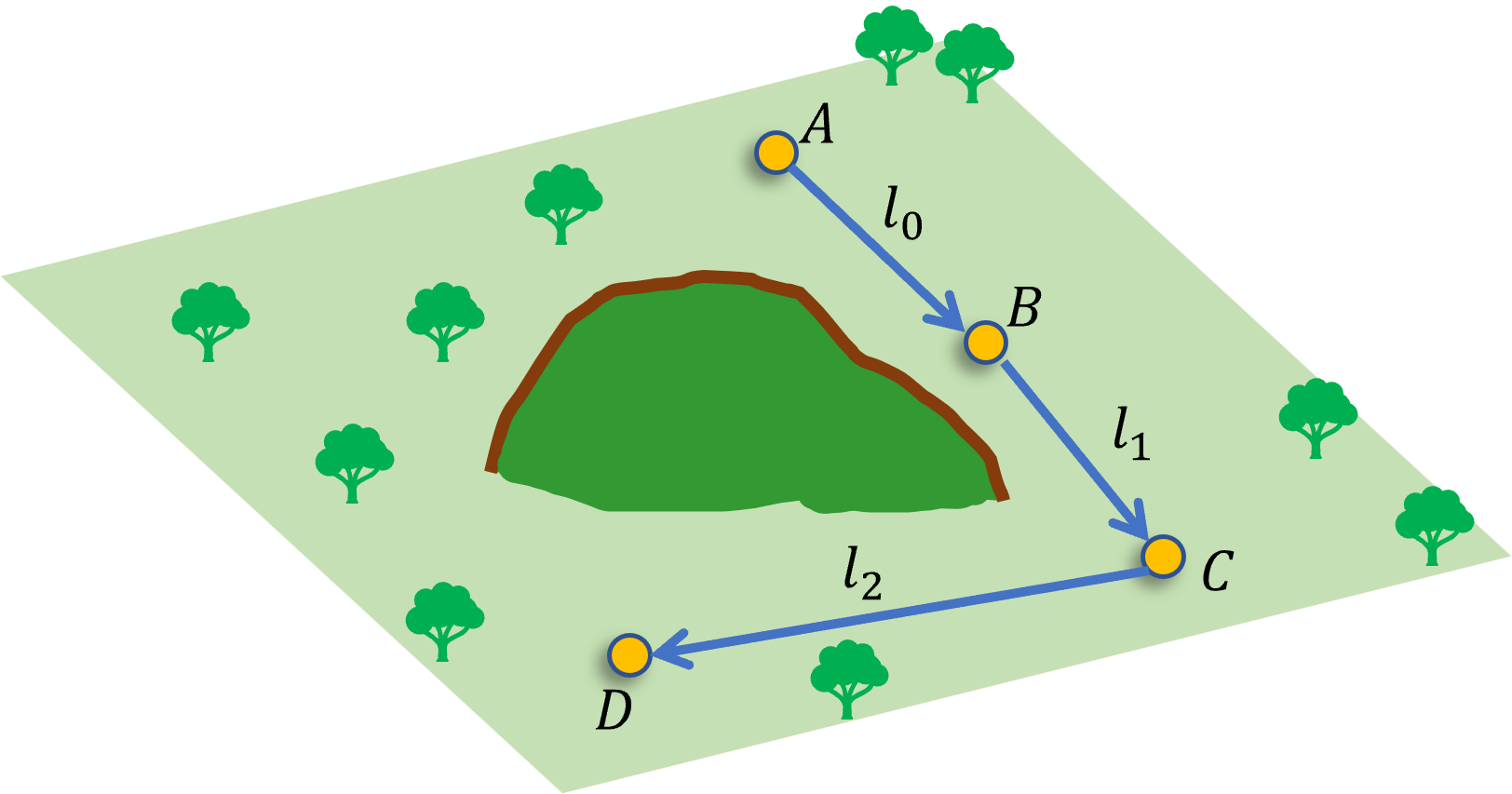}\label{fig:subfig1}}
  \qquad
  \subfloat[Sampled transition system $S_d$ with state space $X_d=\{A,C,D\}$. Key states are  $\{A,C,D\}$ and key transition is $(A,\ell^\prime_0,C)$ where trajectory $\tau_d=A,\ell^\prime_0,C,\ell^\prime_1,D$ contains them.]{\includegraphics[width=0.4\textwidth]{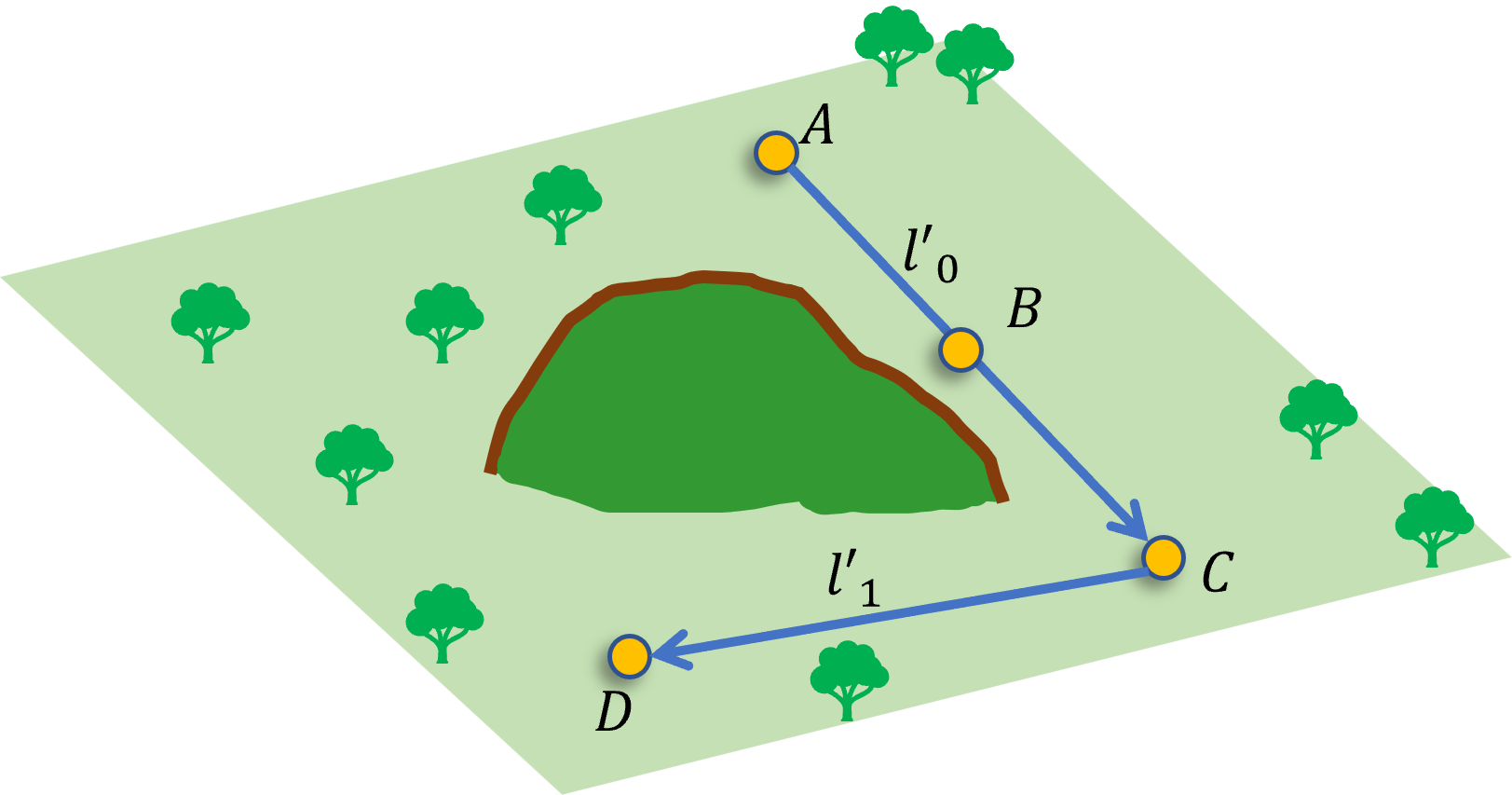}\label{fig:subfig2}}
  \caption{Key states and Key transitions: The state space $X$ of the transition system $S$ contains all key states. However, sampled transition system $S_d$ only contains $A,C,D$. Thus, key state $B$ forms a key transition for $\tau_d$ as $t=(A,\ell^\prime_0,C)$ as there exists an implementation for $\tau_d$ that traverses through $B$.}
  \label{fig:trajectory}
\end{figure*}

For a given task site assignment in the output space of the infinite system, we can find sets of states in a sampled system that would be mapped to the task site assignments. We identify these sets of states as key state classes.

\begin{definition}
[Key State Classes] Let $S_d=(X_d,T_d,L,Y_d,H)$ be a sampled transition system of an infinite state transition system, $S=(X,T,L,Y,H)$. For a given task site assignment $\Xi=\{\hat{Y_1},\hat{Y_2},\dotsc,\hat{Y_{N_y}}\}$, each $\hat{Y_i} \in \Xi$ such that $\hat{Y_i} \cap Y_d \neq \emptyset$ is said to form a key state class $\hat{X}_i$ for $\Xi$ in the sampled transition system $S_d$, i.e., $\hat{X}_i \subset X_d $, where $\hat{X}_i =\{x\in X_d \ |\ \ \exists\  y \in \hat{Y_i}\  \text{s.t.} \  y=H(x) \} $. 
\end{definition}

For a given task site assignment in the output space of the infinite system, it is possible that outputs in some $\hat{Y_i}$ do not exist in the output space of the sampled system, i.e., $\hat{Y_i} \cap Y_d = \emptyset$. Intuitively, this is because some key outputs happen in between samples. For such key outputs, we define key transition classes in the sampled transition system. Figure~\ref{fig:trajectory} illustrates an example instance, where Figure~\ref{fig:subfig1} shows a system that contains all key states and Figure~\ref{fig:subfig2} shows a sampled system that does not contain all key states, thus represents the missing states as key transitions.

\begin{definition}
[Key Transition Classes] Let $S_d=(X_d,T_d,L,Y_d,H)$  be a sampled transition system of an infinite state transition system, $S=(X,T,L,Y,H)$. For a given task site assignment $\Xi=\{\hat{Y_1},\hat{Y_2},\dotsc,\hat{Y_{N_y}}\}$, each $\hat{Y_i} \in \Xi$ such that $\hat{Y_i} \cap Y_d = \emptyset$ is said to form a key transition class $\hat{T}_i$  for $\Xi$ in the sampled transition system $S_d$, such that $\hat{T}_i \subset T_d$.
Any transition $\hat{t}=(x_1,\ell,x_2) \in T_d$ with a label $\ell\in L$ and $x_1,x_2 \in X_d$ belongs to key transition class $i$, i.e., $\hat{t} \in \hat{T}_i $, if there exists an implementation $\zeta_\tau$ for a trajectory $\tau=x_1,\ell,x_2$ and a timestamp $\gamma \in (0,\ell)$ such that $H(\zeta_\tau(\gamma))\in \hat{Y_i} $.
\end{definition}

Now, we have defined key transition classes that arise from the existence of a certain implementation whose output behavior contains task sites that do not exist in the output space of the sampled transition system. Next, we define the specification as the set of trajectories in the transition system that contains states from key state classes and transitions from key transition classes. 

\begin{definition}
[Specification] \label{def:Specifcation} Let $S_d$ be a sampled transition system, $\hat{X}_i$ for $i=0,\dotsc,N_x$ be a set of key state classes, $\hat{T}_j$ for $j=0,\dotsc,N_t$ be a set of key transition classes for a task site assignment $\Xi$. A specification $M$ is a set of execution trajectories of $S_d$, i.e., $M\subset \Sigma(S)$, such that for each trajectory $\tau \in M$, the following holds:\newline(1) For each key state class $\hat{X}_i$ for $i=0,\dotsc,N_x$, there exists a state $x \in \hat{X}_i $ such that $x\in \tau$.
\newline(2)  For each key transition class $\hat{T}_j$ for $j=0,\dotsc,N_t$, there exists a transition $t \in \hat{T}_j $ such that $t \in \tau$.

If an execution trajectory $\tau$ is in $M$, it is said to satisfy the specification.
\end{definition}

% The satisfaction of the mission by a path is decided by the fact that the path contains certain key states. These key states form different classes as specified by the goals of the given mission. We formally define the key state classes as follows.

\subsection{Monotone Specifications}

We define the monotonicity property of a specification to describe the specification satisfaction by trajectories traversing through states that are greater in a defined partial ordering. 
\begin{definition}
[Monotone Specification] \label{def:mono_spec} A specification $M$ for a transition
system $S=(X,T,L,Y,H)$ is said to be a monotone specification iff the following condition is satisfied. Whenever a trajectory $\tau
:=x_{0}\ell_{0}x_{1}\ell_{1}\cdots\ell_{N-1}x_{N} \in M$ 
then any trajectory $\tau^{\prime}
:=x_{0}^{\prime}\ell_{0}x_{1}^{\prime}\ell_{1}\cdots\ell_{N-1}x_{N}^{\prime}$, where $\tau^{\prime}\succeq_{\tau}\tau$, i.e., %
\[
x_{0}^{\prime}\succeq x_{0},x_{1}^{\prime}\succeq x_{1}\cdots,x_{N}^{\prime
}\succeq x_{N}%
\]
also belongs to $M$.\newline
\end{definition}
% (2) For each class of key states $\hat{x}^j$ for $j \in \{0,\cdots,j_{M}\}$, for any $\hat{x}^j_{i}\in \hat{x}^j$ if there exists a state $\hat{x}_{i}^{\prime}$ such that $\hat{x}_{i}^{\prime}\succeq\hat{x}^j_{i}$ then $\hat{x}_{i}^{\prime}$ is a key state of class $\hat{x}^j$ i.e. $ \hat{x}_{i}^{\prime}\in \hat{x}^j $.\newline

% \JDH{Is it possible that a state $x_{i}^{\prime}$ or $\hat{x}_{i}^{\prime}$ is not feasible for a path in $M$? \NT{Yes, it is possible for a mission-satisfying path to exist that does not go through these so-called better states. The idea here is that, if a better path with better key states also satisfies a mission, that mission is called a monotone mission.} So different paths can choose different key states? What makes a state a key state? I assumed that any valid path needs to contain all key states.}

 \color{black}
\subsection{Solver}

We define a solver as an oracle that manifests a trajectory (also referred to as a plan) that could be executed on a given transition system to satisfy a provided specification. For example, a solver can be an optimization engine, a constraint solver, a heuristic-driven algorithm, etc. In general, in addition to the constraints imposed by the transition system and the provided specification, a solver may be assigned a limited amount of computational resources and time.

\begin{definition}
[Solver]
Given a specification $M$ for a transition system $S$, an initial state $x_0$ and a maximum compute time $t^{\max}_{\mathrm{compute}}$, a solver $\Omega$ is an oracle that spends compute time $t^{\Omega}_{\mathrm{compute}}\in(0,t^{\max}_{\mathrm{compute}})$ to search for a specification satisfying execution trajectory $\tau \in M$ and returns it if found. 
\end{definition}
%\AAJ{I think the solver should output an execution trajectory, not just a path.}

\subsection{Objective Function}
While the satisfaction of a specification ensures that the planning goal is achieved, a proposed plan's optimality can be evaluated using an objective function defined as follows.

\begin{definition}[Objective Function] For a monotone transition system $S$ with partial ordering $\succeq$, a function $f_S:\zeta_{\Sigma(S)} \to \mathbb{R} $ is said to be an objective function if for all trajectories $\tau_1,\tau_2 \in \Sigma(S)$ such that $\tau_1 \succeq_\tau \tau_2$, the objective $f_S(\zeta_{\tau_1}) \leq f_S(\zeta_{\tau_2})$.
% \JDH{Does it have to be continuous? \NT{No, I have updated the definition}}
\end{definition}
\color{black}

\color{black}

% For the $j^{th}$ vehicle, we have $x_j=[p^j_{x},p^j_{y},e_j, f_j]$ and the system state is defined as $x=[x_1, x_2,..., x_J,\gamma]$.
% \begin{enumerate}
%     \item The time dimension is sampled with $\gamma_d$ intervals such that a defined resolution is maintained.
%     \item In the spatial dimension, we sample a set of points that are selected by an arbitrary selection process that depends on $\gamma_d$. This process is described in later sections in the context of the running example. 
%     \item In the energy dimension, a set of points is selected depending on $\gamma_d$ such that a suitable resolution is maintained within the energy dimension.
%     \item $f_j$ is a discrete variable by definition.
% \end{enumerate}

\section{Planning Problem in Transition systems}

\subsection{Planning problem}
In this section, we use the defined mathematical concepts to formally state the planning problem in transition systems. Let $S=(X, T, L, Y, H)$ denote an infinite state transition system with a properly defined set of states $X$, transitions $T$, and an output map $H$. Let $\Xi=\{\hat{Y_1},\hat{Y_2},\dotsc,\hat{Y_{N_y}}\}$ denote a provided task site assignment that needs to be satisfied. The goal is to find an implementation $\zeta \colon (0, \gamma) \to X$ of some duration $\gamma \geq 0$, such that the image of its output behavior, $\xi_\zeta \colon (0, \gamma) \to Y$ contains the task sites in the assignment, i.e., $\hat{Y_i} \subset \Ima (\xi_\zeta)$ for $i=1,2,\dotsc,N_y$. Further, an objective function, $f_S \colon \zeta_\Sigma(S) \to \mathbb{R}$, must evaluate the plan optimality when implemented in the system $S$. Therefore, we state the planning problem as an optimization problem as follows.

 \begin{equation}\label{eq:planprob}
 \minimize_{\text{ s.t. } \forall i\  \hat{Y{_i}} \subset \Ima (\xi_\zeta)} f_S(\zeta),
 \end{equation}
 where $\Xi_N $ is the task site assignment.

 Solving the optimization problem in Eq. \eqref{eq:planprob} is computationally expensive due to the infinite state space of the system $S$. Solving for the optimal plan may be infeasible under limited computation resources. Therefore, we define a sampled transition system $S_d=(X_d,T_d,L,Y_d,H)$, creating a discrete state space for a solver to operate. Further, for the task site assignment $\Xi$, we define a set of key state and key transition classes for the sampled transition system $S_d$. Then a specification $M \subset \Sigma(S)$ can be defined using the key state and key transition classes. Now, the problem reduces to finding an execution trajectory $\tau \in M$, such that its implementation $\zeta_\tau$ minimizes the objective function. To this end, we redefine the optimization problem in Eq. \eqref{eq:planprob_discrete} such that solving it yields a specification $M$ satisfying plan in the sampled system.

 \begin{equation}\label{eq:planprob_discrete}
 \minimize_{\tau \in  M} f_S(\zeta_\tau).
 \end{equation}

\subsection{Solution Framework}

\begin{figure}[h]
\centering
\includegraphics[width=0.35\textwidth]{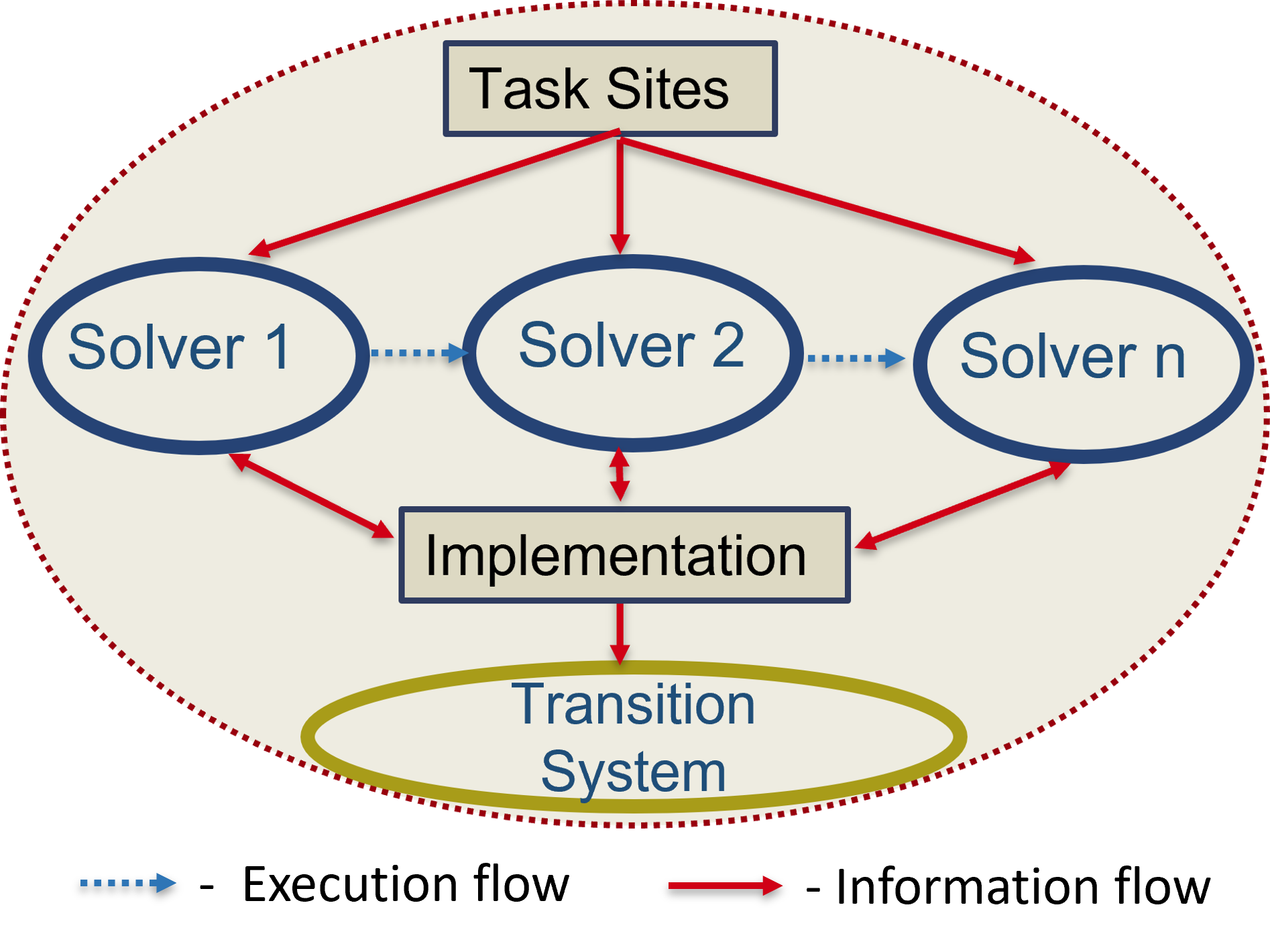}
\caption{Problem space for iterative planning: Solvers are called iteratively with the task site assignment and the current implementation to generate a better plan to be implemented on the transition system.}
\label{fig:probsp}
\end{figure}

Solvers with different capabilities produce solutions with varying levels of sub-optimality using varying computation times. Therefore, we propose a solution framework that enables an iterative solution optimization procedure by creating a problem space that allows multiple solvers to operate and solve the planning problem to formulate better solutions. The proposed framework is shown in Figure~\ref{fig:probsp}. In the problem space, the original infinite state system $S$ and task site assignment $\Xi$ that needs to be satisfied are defined. The goal is to find an implementation in the infinite state system to satisfy $\Xi$. The plan computation is done by invoking each solver iteratively for continual improvement. The proposed framework enables invoking faster solvers earlier to find a feasible solution and then continue to use the available computation resources to optimize the solution further. 

Each solver in the problem space may operate in its own discrete or bounded and continuous state spaces of a sampled transition system $S_d$ and specification $M$ to find an execution trajectory. However, the proposed iterative planning scheme requires that each sampled transition system must agree on the feasibility of the plan provided by the previous solver and its satisfaction of the specification. 

%Further, the final plan at the exhaustion of the computation resources must be executable in the infinite state system.

\subsubsection{Specification Satisfaction}

When a solver $\Omega$ is assigned $t^{\Omega}_{\mathrm{compute}}$ computation time to solve the optimization problem, it will compute a plan $\tau$ as the solution. However, bounded computation time may result in a sub-optimal solution in terms of the objective function. Further, solving the optimization problem in Eq. \eqref{eq:planprob_discrete} yields a plan in the corresponding sampled transition system. Therefore, we state the following theorem  on specification satisfaction.

\begin{theorem}[Specification Satisfaction] \label{theo:task_sat}
Given a sampled transition system $S_d=(X_d,T_d,L,Y_d,H)$ of an infinite state transition system $S=(X, T, L, Y, H)$, and a specification $M \subset \Sigma(S_d)$ for a task site assignment $\Xi$, if a plan(trajectory) $\tau \in \Sigma(S_d)$ satisfies $M$, then its implementation $\zeta$ has the output behavior $\xi_\zeta$ that satisfies the task site assignment $\Xi$.
\end{theorem}

\begin{proof}
For the given task site assignment $\Xi$, let $\hat{X}_i$ for $i=0,\dotsc,N_x$ denote the key states classes and $\hat{T}_j$ for $j=0,\dotsc,N_t$ denote the key transition classes.  
If a plan $\tau$ satisfies $M \subset \Sigma(S_d)$ then $\tau \in M$. From the definition of the specification,  any trajectory $\tau \in M$ must contain some state $x \in \hat{X}_i$ for all $i=0,\dotsc, N_x$ and some transition $t \in \hat{T}_j$ for all $j=0,\dotsc,N_t$.
From the definition of key states and key transitions, for an implementation $\zeta$ of any trajectory $\tau \in M$, there exists a $\gamma \geq 0$ such that $H(\zeta(\gamma))\in \hat{Y}_i$ holds for all  $\hat{Y}_i \in \Xi$. 
Therefore, the output behavior of $\zeta$ contains the task sites in its image, i.e., $\hat{Y_i} \subset \Ima (\xi_\zeta)$ for $i=0,\dotsc,N_x+N_y$. Thus, implementation of any trajectory $\tau \in M$ yields an output behavior $\xi_\zeta$, satisfying the task site assignment $\Xi$.
\end{proof}

% \begin{figure}[h]\label{fig:wrapperframework}
% \centering
% \includegraphics[demo,width=2cm,height=2cm]{example-image} \caption{}
% \end{figure}

Therefore, any solver that operates on a sampled transition system can find solutions that satisfy the task site assignment $\Xi$ when implemented in the infinite state system. 

\subsubsection{Ability of Improvement}

 As the solvers operate on a sampled transition system $S_d=(X_d,T_d,L,Y_d,H)$, with a bounded computation time, a plan produced by the solvers could be sub-optimal in terms of the objective function $f_S$. Therefore, for a produced trajectory $\tau$, if an improvement $\tau ^\prime \succeq_\tau \tau$ exists such that  $x_i^\prime \in \tau_d$ for any state $ x_i \in \tau$ where $x_i^\prime\succeq x_i$, the transitions in $\tau^\prime$ must exist in $T_d$ under the replacement of $x_i$ by  $x_i^\prime$. Thus, for any $t_ij=(x_i,\ell,x_j)\in \tau$ and $x_i^\prime\succeq x_i$, a $t_{ij}^\prime=(x_i^\prime,\ell,x_j)\tau^\prime$ must exist in $T_d$. %Any plan that satisfies the specification must be improvable by choosing better states in terms of the partial ordering $\succeq$, if such states exist in $X_d$.  
 To this end, we state the following theorem on the monotonicity of the sampled transition systems.

\begin{theorem}[Monotonicity in Sampled Transition System]\label{theo:mono}
Given a monotone transition system  $S=(X, T, L, Y, H)$ with respect to a partial ordering $\succeq$, a sampled transition system $S_d=(X_d,T_d,L,Y_d,H)$ of $S$ is also monotone with respect to $\succeq$.
\end{theorem}

\begin{proof}
For a transition system $S=(X, T, L, Y, H)$, the monotonicity condition states that the following is true for every
$x_{1},x_{1}^{\prime}\in X$ and $\ell\in L$, such that $(x_{1},\ell,x_{1}^{\prime}) \in T$: For all $x_{2}\in X$ such that $x_{2}\succeq x_{1},$ there exists a transition $(x_{2},\ell,x_{1}^{\prime})\in T$.
From the definition of sampled transition system $S_d=(X_d,T_d,L,Y_d,H)$, it is known that $X_d \subset X$, and $T_d \subset T$. For all $x_{1},x_{1}^{\prime}, x_{2}\in X_d$ and  $\ell\in L$, such that $(x_{1},\ell,x_{1}^{\prime}) \in T_d$ and $x_{2}\succeq x_{1}$, a transition $t^d=(x_{2},\ell,x_{1}^{\prime})$ exists as $x_{1},x_{1}^{\prime},x_{2} \in X$ where $S=(X, T, L, Y, H)$ is monotone. As the transition, $t^d$ exits in $T_d$, the sampled transition system $S_d$ is monotone.
\end{proof}
\color{black}

% \begin{theorem}
% Given a monotone transition system  $S=(X, T, L, Y, H)$ with respect to a partial ordering $\succeq$, a monotone specification $M$, a sampled transition system $S_d=(X_d,T_d,L,Y_d,H)$ of $S$, the discrete mission $M_d \subset M$ is also monotone.
% \end{theorem}

% \begin{proof}
% Let a path $\xi \in \Xi(S_d)$ be a mission-satisfying path such that $\xi \in M_d$. Let $\xi^\prime \in \Xi(S_d) $ be a path that is higher in partial ordering such that $\xi^\prime \succeq_\xi \xi$. Form the monotonicity of mission $M$, $\xi^\prime \in M $. However, it is know that $M_d = M \cap \Xi(S_d)$. Therefore $\xi^\prime \in M_d$. Thus, the discrete mission $M_d$ is monotone.
% \end{proof}

\subsection{Iterative planning}

\begin{algorithm}
\RestyleAlgo{ruled}
\caption{Iterative Planning}\label{alg:planning}
% $t^{i}_{compute}$- Max compute time for solver $i$, $\delta_f(<1)$ - Objective improvement factor
\KwData{$S$- Infinite state transition system, $\Xi$ - Task site assignment, $x_0$ - Initial state, $(\Omega_1,\dotsc,\Omega_n)$ - Sequence of solvers,  $t^{total}_{compute}$- Max total compute time,  $\gamma_{H}$- Horizon length}
\KwResult{$\zeta_{final}$- Final plan implementation}
$i \gets 1$\;
 Set $S_{d_i}$ to a sampled transition system of $S$ \;
 Define specification $M \in \Sigma(S_{d_i})$ from $\Xi$\;
 Get current time $t_{start}=clock()$\;
 $t^{i}_{compute}\gets t^{total}_{compute}+ t_{start}-clock()$\;
 \If{$t^{i}_{compute}>0$}{
    Call $\Omega_i$ with compute time $t^{i}_{compute}$ to find a plan $\tau_{d_i} \in M$ starting from $x_0$ for a horizon of $\gamma_{H}$\;
    Set $\zeta_i\gets\zeta_{\tau_{d_i}}$\;
}

\For{$i=2,\dotsc,n$}{
 Define $S_{d_i}$  such that $ \{ x\ | \  \exists \gamma\in \Gamma \ x=\zeta(\gamma)\} \subseteq X_{d_i} $ where the set of time stamps $\gamma^d=\gamma^d_0,\dotsc,\gamma^d_N$ \;
 Define the specification $M \in \Sigma(S_{d_i})$ from $\Xi$\;
 $t^{i}_{compute}\gets t^{total}_{compute}+ t_{start}-clock()$\;
 \If{$t^{i}_{compute}>0$}{
Call $\Omega_i$ with compute time $t^{i}_{compute}$ to find a plan $\tau_{d_i} \in M$ starting from $x_0$ for a horizon of $\gamma_{H}$\;
    \If{$ f_S(\zeta_{\tau_{d_i}}) \leq f_S(\zeta_{i-1}) $}
    { Set $\zeta_i\gets\zeta_{\tau_{d_i}}$\;
    }
    \Else    {Set $\zeta_i\gets \zeta_{i-1}$\;
    }
}
\Else    {break\;
    }
 % \If{$clock() + t^{i+1}_{compute} \geq t^{total}_{compute}+t_{start}$}
 %    {break\;
 %    }
}
\Return $\zeta_{final}=\zeta_i$\;

\end{algorithm}

In the proposed iterative planning framework, the solvers work on finding a plan satisfying the specification. Each solver takes in the specification and creates a sampled transition system instance in the solver space. Then, it computes the solution as a trajectory and returns it to the infinite state system as an implementation. When the next solver is called for plan computation, it must identify the current implementation $\zeta$ that satisfies the task site assignment and corresponds to a specification satisfying trajectory in the solver's sampled transition system. To this end, we state the following theorem on the feasibility of the current implementation through iterative planning instances.

\begin{theorem}[Recursive feasibility]\label{theorem:feasibility}
    Given a monotone transition system  $S=(X, T, L, Y, H)$ and an implementation $\zeta$, whose output behavior satisfies a task site assignment $\Xi=\{\hat{Y_1},\hat{Y_2},\dotsc,\hat{Y_N}\}$, let $S_d=(X_d,T_d,L,Y_d,H)$ be a sampled transition system of $S$, and $M$ be a specification defined from $\Xi$. For any set of time stamps $\Gamma=\{\gamma_0,\dotsc,\gamma_\mu\}$, if $ \{ x\ | \  \exists \gamma\in \Gamma x=\zeta(\gamma)\} \subseteq X_d $, then there exists a trajectory $\tau \in \Sigma(S_d)$ such that $\tau$ satisfies $M$.
\end{theorem}

\begin{proof}
For a given transition system  $S=(X, T, L, Y, H)$, and any implementation $\zeta$, let $x_i=\zeta(\gamma_i) \in X$ for  $\gamma_i \in \Gamma$. Let time durations $\gamma^d_i=(\gamma_{i+1} - \gamma_i)\in L$ for $i\in\{0,\dotsc,\mu-1\}$.  Then a corresponding trajectory to $\zeta$ in $S$ can be defined as $\tau
=x_{0}\gamma^d_0 x_{1}\gamma^d_1 \cdots \gamma^d_{\mu-1}x_{\mu} \in \Sigma(S)$. However, for a sampled transition system $S_d=(X_d,T_d,L,Y_d,H)$ sampled from $S$, the existence of $\tau \in \Sigma(S_d) $ requires that each $x_i\in \tau$ must be in $X_d$. Therefore, if $ \{ x\ | \  \exists \gamma\in \Gamma x=\zeta(\gamma)\} \subseteq X_d $ for some  $\Gamma$, we see that there exists a $\tau \in \Sigma(S_d)$. 

If $M$ is a specification in $S_d$ defined from $\Xi$, any $\tau \in M$ must contain states from each key state class and transition from each key transition class.
From the definition of key state classes, let a subset of task site allocation, $\Xi_x=\{\hat{Y}_{j}\in \Xi \ |\ \hat{Y}_{j}\cap Y_d \neq \emptyset \}$ form key state classes in $S_d$.
For a given implementation $\zeta$ that satisfies a task site assignment $\Xi$, there exists some $\gamma_j \geq 0$ such that $H(\zeta(\gamma_j)) \in \hat{Y}_{j}$ for all $j=1,\dotsc,N$. Let $\Gamma$ be defined such that $\{\gamma_j \ |\  \exists \  \hat{Y}_{j}\in \Xi_x H(\zeta(\gamma_j)) \in \hat{Y}_{j}\  \} \subseteq \Gamma$. For $\gamma_i\in \Gamma$ and time durations $\gamma^d_i=(\gamma_{i+1} - \gamma_i)\in L$ for $i\in\{0,\dotsc,\mu-1\}$, a corresponding trajectory to $\zeta$ defined as $\tau
=x_{0}\gamma^d_0 x_{1}\gamma^d_1 \cdots \gamma^d_{\mu-1}x_{\mu}$ contains states from each key state class, i.e.  there exists $\gamma_j \in \Gamma$, such that $\zeta(\gamma_j)\in \tau$ and $H(\zeta(\gamma_j)\in \hat{Y}_{j}$ for all $\hat{Y}_{j}\in\Xi_x$.

% \JDH{I don't understand this first sentence. Do the task sites form key transitions? Or is it $i_t$ that forms key transitions?}
% \NT{The task sites forms key transitions. I have edited the whole proof to make it more clear.}

Let $\Xi_t =\Xi \backslash \Xi_x$ be the set of task sites that do not form key state classes in $S_d$, i.e., $\Xi_t=\{\hat{Y}_{j}\in \Xi \ |\ \hat{Y}_{j}\cap Y_d = \emptyset \}$. Therefore, they form key transition classes. From the given existence of implementation $\zeta$ that satisfies $\Xi$, there exists some $\gamma_j \geq 0$ such that $H(\zeta(\gamma_j))\in \hat{Y}_{j}$ for all $\hat{Y}_{j}\in \Xi_t$. Therefore, a corresponding trajectory to $\zeta$ defined as $\tau
=x_{0}\gamma^d_0 x_{1}\gamma^d_1 \cdots \gamma^d_{\mu-1}x_{\mu}$ contains transitions from each key transition class by definition. 
Since $\tau \in \Sigma(S_d)$ contains states from each key state class and transitions from each key transition class, it satisfies $M$. Thus, a trajectory $\tau \in \Sigma(S_d)$ exists such that  $\tau
=x_{0}\gamma^d_0 x_{1}\gamma^d_1 \cdots \gamma^d_{N-1}x_{N}$ satisfies $M$.
\end{proof}

\color{black}
Theorem \ref{theorem:feasibility} establishes that given the set of states, $X_d$ of a sampled transition system includes states sampled from an implementation  that satisfies a task site assignment, there exists a specification-satisfying trajectory in the sampled transition system of the solver. This ensures the existing plan remains feasible for each sampled transition system when solvers are iteratively called through the planning framework. 

From Theorem \ref{theorem:feasibility}, each solver in the iterative planning framework can compute solutions while the current implementation remains feasible. Each solver employed in the framework can be invoked to generate solutions with better objective values. To this end, we propose Algorithm \ref{alg:planning}. 

Iterative planning takes the transition system $S$, task site assignment $\Xi$, initial state $x_0$ and a set of solvers $\{\Omega_1,\Omega_2, \dotsc, \Omega_n\}$. Each solver is allowed the maximum computation time. At each step of planning in the iterative framework, a sampled transition system $S_d$ and a specification are created such that Theorem \ref{theorem:feasibility} is satisfied. 

\begin{corollary}
    Let $S$ be a transition system and $\Xi$ be a task site assignment. Let $\Omega_i$ s.t. $i>1$ be a solver in Algorithm \ref{alg:planning} and $S^i_d$ be the corresponding sampled transition system. For an implementation $\zeta_{i-1}$ found from the previous solver $\Omega_{i-1}$, there exists an implementation $\zeta_{i} \in \zeta_{\Sigma(S^i_d)}$ whose output behavior satisfies $\Xi$. 
\end{corollary}

\begin{proof}
    This is a direct consequence of Theorem \ref{theorem:feasibility}. In Algorithm \ref{alg:planning}, each $S^i_d$ is sampled such that \ref{theorem:feasibility} is satisfied. It guarantees the existence of a $\tau \in \Sigma(S^i_d)$ such that the output behavior of its implementation $\zeta_{i} \in \zeta_{\Sigma(S^i_d)}$ satisfies $\Xi$. 
\end{proof}

This ensures the recursive feasibility of the planning framework. Then, the corresponding solver is allowed the given computation time for plan computation. If the computed plan is better than the current plan, the current plan is updated to the computed plan. Once all the solvers are iterated or the allowed compute time runs out, the current plan is returned for implementation. Note that in Algorithm \ref{alg:planning}, $clock()$ indicates the Timer function.

The planning scheme described in Algorithm \ref{alg:planning} has several advantages compared to the direct execution of a plan from a single solver.
\begin{enumerate}
    \item Multiple solver integration - 
    Provided that the corresponding sampled transition system satisfies Theorem \ref{theorem:feasibility}, any solver can be interfaced seamlessly with the planning framework. This enables a versatile implementation that can iterate over solvers with varying sampled state spaces. 
    \item Compute time constraints - Allocation of compute time for each solver as well as for the total planning scheme allows time-sensitive planning tasks to be completed with the planning framework.
    \item Continual improvement - At each iterative call to the solvers, an improvement constraint is imposed on the next solution. This enables continuous improvement of the solution toward optimality. 
    \item Implementable solution - At each step of planning, an implementation in the original system 
  that satisfies the task site assignment is always at hand. If the next solver called iteratively cannot provide a better plan, the existing plan can be directly executed on the infinite state system to satisfy the task site allocation.
\end{enumerate}

Next, we describe the shrinking horizon execution followed for the plan produced by Algorithm \ref{alg:planning}.

\subsection{Shrinking horizon execution}

\begin{figure*}[t]
\centering
\includegraphics[width=0.7\textwidth]{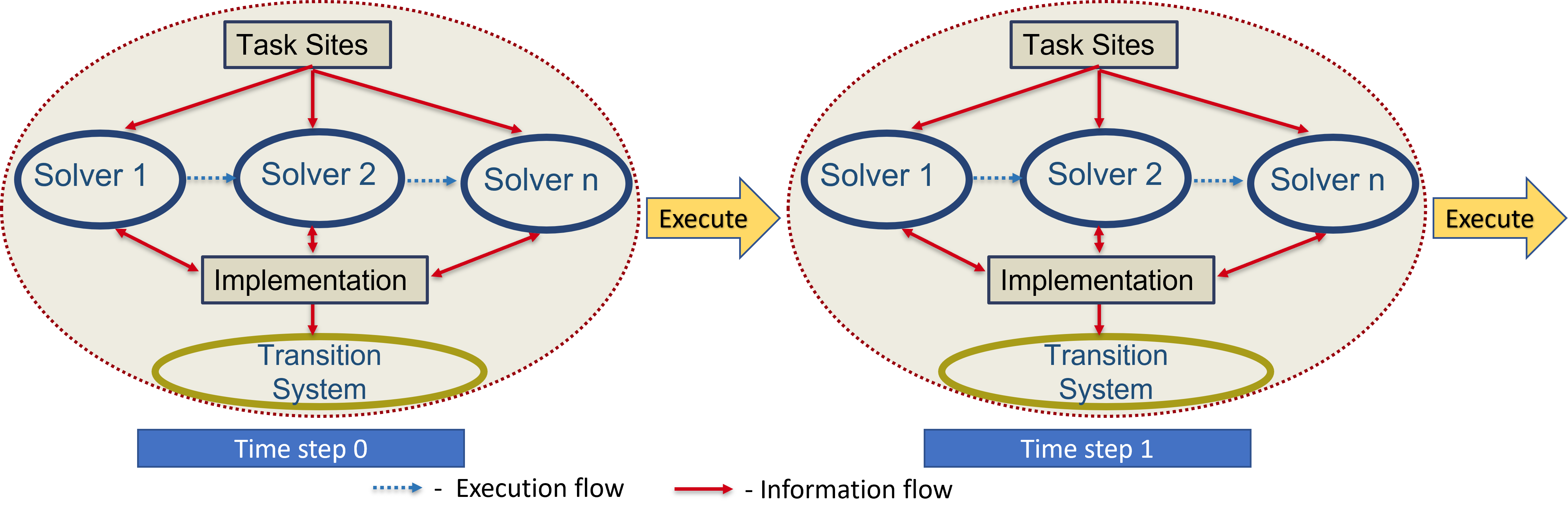}
\caption{Shrinking horizon execution: At each step of execution, a problem space is created to run iterative planning.}
\label{fig:Execution_scheme}
\end{figure*}

After each iterative planning step, the produced plan is executed step by step in a shrinking horizon fashion following the execution scheme shown in Figure \ref{fig:Execution_scheme}. Algorithm \ref{alg:execution} describes the proposed execution scheme. At each planning step, Algorithm \ref{alg:planning} is called to compute a plan for the current specification. Once a plan with a better objective is computed, the plan is executed a time step forward. If a task site $\hat{Y}_j$  was visited at the execution of the current step, the assignment is updated for the next step by excluding $\hat{Y}_j$ from the task site assignment. As Theorem \ref{theorem:feasibility} ensures the recursive feasibility of the current plan, it can always be executed in the system if a better plan does not exist. In the next section, we describe an example implementation of the proposed solution framework.

\begin{algorithm}
\RestyleAlgo{ruled}
\caption{Shrinking Horizon Execution}\label{alg:execution}
\KwData{$S$-Infinite state transition system, $\Xi_0$ - Task site assignment, $x_0$ - Initial state, $k\in(1,\dotsc,K)$ - Execution time steps, $t^{total}_{compute}(<\gamma_{step})$- Max total compute time for each step,  $\gamma_{step}$ - Duration of a time step.}
\KwResult{$\zeta_{final}$- Executed implementation}
\For{$k=1,\dotsc,K$}{
Create \textit{Iterative Planning} instance $k$ using system $S$ and task site assignment $\Xi_{k-1}$\;
Set initial state $x_{k-1}$\;
Run \textit{Iterative Planning} for $\Xi_{k-1}$ on system $S$ starting from $x_{k-1}$ for a horizon of $\gamma_{H}=(K-k+1)\times \gamma_{step}$ with compute time $t^{total}_{compute}$\;
Get $\zeta^*$ from \textit{Iterative Planning} output\;
\If{$k= 1 \ \lor \ f_S(\zeta^*) \leq f_S(\zeta_{k-1})$}{
$\zeta_{k} \gets \zeta^* $\;}
\Else{
$\zeta_{k} \gets \zeta_{k-1}$\;
}
Set $x_{k}=\zeta_{k}(\gamma_{step})$ after execution for $\gamma_{step}$ duration \;
Let $\hat{Y}_j$  for $j\in \{1,\dotsc,j_{k-1}\}$ be task sites of $\Xi_{k-1}$\;
Get satisfied task sites,\newline 
$\Phi=\{\hat{Y}_j \in \Xi_{k-1}\ |\ \exists \gamma \leq \gamma_{step} \text{ s.t } H(\zeta_{k}(\gamma))\in \hat{Y}_j\}$\;
Shift $ \zeta_{k}$ by $\gamma_{step}$, i.e., $\zeta_{k}(\gamma) \gets \zeta_{k}(\gamma - \gamma_{step})$\;
Exclude satisfied task sites and update assignment, $ \Xi_k= \Xi_{k-1} \backslash \Phi$\;
}
\Return $\zeta_{final}=\zeta_{K}$\;
\end{algorithm}

%\section{Cooperative Vehicle Route Planning Problem}

% \begin{figure}[h]\label{scenario}
% \centering
% \includegraphics[width=0.45\textwidth]{Picture1.png}
% \caption{Example problem scenario}
% \end{figure}

\color{black}
\section{Application: Energy-aware Route Planning}

In this section, we introduce an example planning problem and solve it using the techniques of this paper. Consider a scenario where a collaborative team of vehicles is assigned a set of tasks to be completed at different locations on a known map. The objective is to find the optimum route plan for each vehicle such that the overall plan minimizes the total time to complete all assigned tasks. This problem can be viewed as a VRP, an extension of the TSP. However, complexities arise when the vehicles are of different types, like UAVs or UGVs with different capabilities. Further, when the vehicles have limited energy, the route planner must account for energy consumption and properly handle recharging. Therefore, this example represents a generalized energy-aware cooperative task site assignment.

\color{black}
\subsection{Running Example} \label{sec:example}

\begin{figure}[h]
\centering
\includegraphics[width=0.45\textwidth]{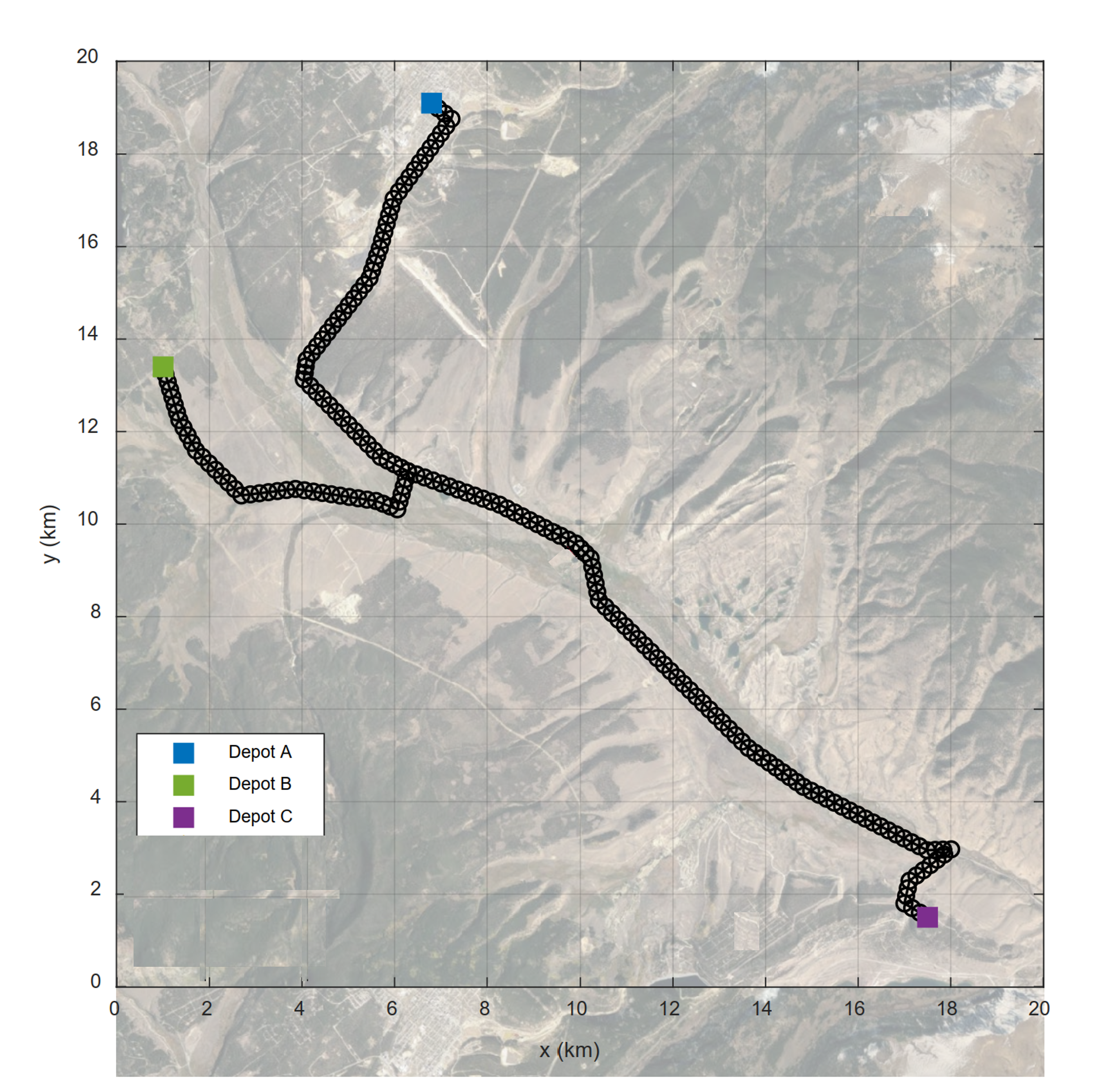}
\caption{Example map.}
\label{fig:map}
\end{figure}

% We introduce the example application in energy-aware UAV-UGV
% cooperative task site assignments.

% \JDH{[This sentence is unclear. Do you mean ``example used in this paper, which we will solve and analyze the resulting route''?]}. 
Consider a known 2D map of an area of interest, $G=[0,20]\times[0,20] $.
%\{(x,y)\in \mathbb{R}^2 \  | \  \exists x \in [0,20], \ \exists y \in [0,20] \}$ of an area of interest. 
Let $G_R \subset G$ denote a road and three depot locations at the ends as in Figure \ref{fig:map}. The task is to monitor the road using a team of UAVs and a UGV, where each vehicle has the required perception abilities. Also, the vehicles have limited energy and can be recharged if required. The depots, \textit{A, B} and \textit{C}, are equipped with charging pads for wireless charging of UAVs and a battery-swapping facility for the UGV. Further, the UGV is also equipped with a charging pad to conduct wireless charging of two UAVs at a time. We describe each vehicle in detail as follows.

\subsubsection{UGV specifications}

% \JDH{This is a formatting request. When in math mode, use \textbackslash mathrm\{\} for multi-letter labels, words, and units, like ``max''. Leave variables as-is (italicized by default). For example, I've changed ``max, UGV, MJ, and ms'' to mathrm in the second sentence below. You can even go a step further and use the package siunitx for the units, which I think handles units well.}

The motion of the UGV is restricted to the road $G_R$ denoted in black in Figure \ref{fig:map}. It has a total usable energy of $B^{\mathrm{UGV}}=25.01 \,\mathrm{MJ}$. The maximum velocity, $v_{g_{\mathrm{max}}} = 4.5 \,\mathrm{ms}^{-1}$. The power consumption curve of the UGV in Watts is as follows,
\begin{equation}\label{eq:gndpower}
      Pe_g(v_g)=
\begin{cases}
  200\mathrm{W}, \qquad \qquad \qquad  \qquad v_g=0,\\
  1.05(464.8V_g +356.3)\mathrm{W}, \ \ v_g>0.\\
\end{cases}
\end{equation}

\subsubsection{UAV specifications}

The UAVs can travel freely within the map $G$. They have limited energy storage onboard with a total of $B^{\mathrm{UAV}}=287.7\ \mathrm{kJ}$ usable energy. The maximum UAV velocity is $v_{a_{\mathrm{max}}}=16\ \mathrm{ms}^{-1}$. The power consumption curve of a UAV in Watts is as follows,
\begin{equation}\label{eq:airpower}
      Pe_a(v_a)=1.05(0.0461v_a^3-0.5834v_a^2- 1.8761v_a+ 229.6)\mathrm{W}.
\end{equation}

UAVs can dock for recharging on a UGV or a charging pad at a depot. The wireless power transfer rate in Watts when recharging follows the following curve, which depends on the UAV battery's current energy level ($E_a$ in kJ).

\begin{equation}\label{eq:aircharge}
      Pe_{\mathrm{chg}}(E_a)=
\begin{cases}
  310.8, \qquad \qquad \qquad  \qquad E_a \leq 270.4,\\
  17.965(287.7-E_a), \ 270.4 \leq E_a \leq 287.7\\
\end{cases}
\end{equation}

We use the example map and the vehicle specifications as the selected application in this work. Each UAV can travel to any point in the 2D map $G$, and each UGV travels only along the road $G_R$. Further, the energy is consumed or replenished following the consumption and charging curves. The continuous states of spatial, energy, and time increase the complexity of the planning problem. When the number of vehicles increases, the number of state and search space dimensions increases. Planning problems like the running example involve searching for solutions in multi-dimensional infinite-state systems. In this work, we apply the defined concepts of transition systems to mathematically represent the infinite state systems for the motion planning task. In the next section, we formulate the routing problem mathematically using defined preliminaries. Further, we analyze the results of the proposed solution when the problem is scaled in terms of the number of agents and the number of task sites.

\color{black}

\subsection{Transition System Formulation}\label{sec:origsys}

In this section, we mathematically formulate the transition system for the route planning problem in the running example scenario using the established theoretical framework. Let the team consist of $J$ vehicles. We denote the 2D spatial coordinate of the $j^{th}$ vehicle as $(p^j_{x},p^j_{y})\in G $, and the energy level as $e_j\in [0,B^j_{max}]$ where $B^j_{max}$ is the battery capacity of vehicle $j$. Further, we define status flag $f_j$ encoding the vehicle type (UGV/UAV) and its status. The value of $f_j$ is assigned as shown in Table \ref{table:f_i}.

\begin{table}[h!]
    \caption{Values of the status flag $f_j$}
\begin{center}
\begin{tabular}{ | m{1em} | m{2cm}| m{3cm} | } 
  \hline
  $f_j$ & Vehicle type & Status \\ 
  \hline
  0 &  UGV & No UAVs docked \\ 
  \hline
  1 & UGV & 1 UAV docked \\ 
  \hline
    2 &  UGV & 2 UAVs docked \\ 
  \hline
    3 &  UAV & Not docked \\ 
  \hline
    4 &  UAV & Docked \\ 
  \hline
\end{tabular}
\end{center}
\label{table:f_i}
\end{table}

Then, we define the state of the multi-vehicle system at time $\gamma$ as follows.
\begin{equation}\label{eq:infinite_system}
    x=[x_1, x_2,..., x_J,\gamma]\in X,
\end{equation}
where $x_j=[p^j_{x},p^j_{y},e_j, f_j]$ for $j=1,2,...,J$ and $X$ is the state space. We define a transition system for using the state and transitions representing state evolution as follows. 

\begin{definition}
[Multi-vehicle System]   \label{def:fulsys}
Given a multi-vehicle system of $J$ vehicles whose state is denoted as $x_j=[p^j_{x},p^j_{y},e_j, f_j] $ for $j=1,2,...,J$, a transition system $S=(X,T,L,Y,H)$ is defined as follows.\newline
(1) The state space, $X=\{x\ |\  x=[x_1, x_2,..., x_J,\gamma] \}$ where, $x_j=[p^j_{x},p^j_{y},e_j, f_j] $ for $j=1,2,...,J$  and $\gamma$ is the time.\newline
(2) The set of transitions, $T=\{t\ |\ \exists \delta_\gamma > 0 \text{, s.t. } t=(x,\delta_\gamma,x^\prime) \text{ for } x,x^\prime \in X\}.$\newline
(3) The set of labels, $L=\{\delta_\gamma > 0 \ |\ \exists t \in T \text{, s.t. } \  t=(x,\delta_\gamma,x^\prime) \text{ for } x,x^\prime \in X \}.$\newline
(4) Output space,  $Y=\{y\ |\  y=[y_1, y_2,..., y_J] \}$ where, $y_j=[p^j_{x},p^j_{y}] $ for $j=1,2,...,J$.\newline
(5) Output map, $H\colon X \to Y$.
\end{definition}

Here, any transition $t=(x,\delta_\gamma,x^\prime) \in T$ is defined as an evolution of state from any time $\gamma$ to $\gamma+\delta_\gamma$, where vehicles can change spatial coordinates, energy levels, and status, flags $f_j$, under their dynamics and power consumption curves. Each transition is assigned a label $\delta_\gamma \in L$. We set the output $y=H(x)=[y_1, y_2,..., y_J]$, where each $y_j=[p^j_{x},p^j_{y}]$, which contains the spatial coordinates of the $j^{th}$ agent of the system. Now, the transition system for the infinite state system is fully defined as $S=(X,T,L,Y,H)$. Next, we discuss some required properties of the defined system.

\subsection{Ordering Relation and Monotonicity}

For a system $S=(X,T,L,Y,H)$, we introduce the partial ordering relation $\succeq$ as a measure of a state's goodness for comparison as follows. 
\begin{definition}[Partial Ordering]\label{def:partialordering}
    For a transition system $S=(X,T,L,Y,H)$ defined from a multi-vehicle system of $J$ vehicles following \ref{def:fulsys}, let any state be defined as,  $x=[x_1, x_2,..., x_J,\gamma]\in X$ where, $x_j=[p^j_{x},p^j_{y},e_j, f_j] $ for $j=1,2,...,J$  and $\gamma$ is the time. A binary relation $\succeq$ that is reflexive, anti-symmetric, and transitive is called partial ordering relation defined on $X$, if it satisfies the following. \newline
    (1) For any $x,x^\prime \in X$ such that $x^\prime \succeq x$, $H(x)=H(x^\prime)$ must hold. This ensures that the spatial coordinates of all agents in $x$ and $x^\prime$ are equal when comparing the energy and time dimensions.
    %\JDH{what does this constraint mean?}
    \newline
    (2) For any states $x,x^\prime \in X$, let $e_j, e^{\prime}_j  $ denote energy levels of any vehicle $j\in \{1,\dotsc,J\}$, respectively. If $x^\prime \succeq x$ then $e^{\prime}_j \geq e_j$ for all $j \in \{1,\dotsc,J\}$ \newline
    (3) For any $x,x^\prime \in X$, let $\gamma, \gamma^{\prime}$  denote time, respectively. If $x^\prime \succeq x$, then  $\gamma^{\prime} \leq \gamma$.\newline
\end{definition}

Given the partial ordering, the system states can now be ranked by measuring their goodness using the energy level of the vehicles and the time. Further, with the partial ordering defined, the monotonicity of $S$ follows from the Definition \ref{def:monotonesys}. 

\begin{corollary}\label{Cor:monotone}
A transition system $S=(X,T,L,Y,H)$ defined from a multi-vehicle system of $J$ vehicles following Definition \ref{def:fulsys} is monotone with respect to a partial ordering $\succeq$ as defined in Definition \ref{def:partialordering}.
\end{corollary}

\begin{proof}
Let any $x_{1},x_{1}^{\prime}\in X$ and $\delta_\gamma \in L$, such that a transition $t=(x_{1},\delta_\gamma ,x_{1}^{\prime}) \in T$ exists. Let $x_{2}\in X$ such that $x_{2}\succeq
x_{1}$, that is $x_{2}$ ranks better than $x_{1}$ due to having higher energy level or earlier time. Because it is given that a transition $t=(x_{1},\delta_\gamma ,x_{1}^{\prime})$ exists,  $x_{1}^{\prime}$  is reachable from $x_{1}$ in time $\delta_\gamma$. Therefore, $x_{1}^{\prime}$  is reachable from a state with higher energy or earlier time than $x_{1}$, like $x_{2}\in X$. Thus, for any $x_{2}\in X$ such that $x_{2}\succeq x_{1}$, a transition $(x_{2},\ell,x_{1}^{\prime})\in T$ always exists. Therefore, the transition system $S$  defined from a multi-vehicle system following Definition \ref{def:fulsys} is monotone.
\end{proof}

\subsection{Task Site Assignment for Monitoring Tasks}

In the running example, the goal is to monitor the road map $G_R$. Each road location forms classes in the output space $Y$ to create the task site assignment. 

\begin{definition}[Task Site Assignment for Monitoring Tasks]
    For a transition system $S=(X,T,L,Y,H)$ defined from a multi-vehicle system of $J$ vehicles following Definition \ref{def:fulsys}, and a road $G_R$ of $N_R$ locations, each location $\hat{y_i}=(p^i_{x},p^i_{y}) \in G_R $ forms an output class $\hat{Y}_i$ in the task site assignment such that $\hat{Y}_i=\{y\ |\ \exists y \in Y \text{ s.t. } \hat{y_i} \in y \text{ where } \hat{y_i} \in G_R \}$. The collection of all output classes is denoted as the task site assignment. $\Xi=\{\hat{Y}_1,\dotsc, \hat{Y}_{N_R}\}$.
\end{definition}

Thus, output behavior $\xi_\zeta$ of any implementation $\zeta$, such that $\Ima{\xi_\zeta} \cap \hat{Y}_i \neq \emptyset$ for all $i=1,\dotsc,N_R$, satisfies the goal of road monitoring.

\subsection{Planning Problem for Monitoring Tasks}

Each implementation that satisfies the task site assignment is evaluated using the objective function. The goal of the planning task is to find a trajectory that satisfies the task site assignment when implemented while minimizing the total time spent executing the trajectory. Therefore, we define the objective function for monitoring tasks as follows. 

\begin{definition}[Objective Function for Multi-vehicle System]
For a transition system $S=(X,T,L,Y,H)$ defined from a multi-vehicle system of $J$ vehicles following Definition \ref{def:fulsys}, let $\tau
=x_{0}\delta_\gamma^{0}x_{1}\delta_\gamma^{1}\cdots\delta_\gamma^{N-1}x_{N} \in \Sigma(S)$ denote a trajectory and $\zeta_\tau$ be its implementation. An objective function $f_S \colon \zeta_{\Sigma} \to \mathbb{R}$, is defined to return the time of the final state, i.e., $\gamma_N=f_S(\zeta_\tau)$ where $\gamma_N\in x_N$. 
\end{definition}

The goal is to find an implementation $\zeta \colon (0, \gamma) \to X$ of some duration $\gamma \geq 0$, such that the image of its output behavior, $\xi_\zeta \colon (0, \gamma) \to Y$ contains the task sites in the assignment, i.e., $\hat{Y_i} \subset \Ima (\xi_\zeta)$ for $i=1,2,\dotsc,N_R$. The objective function, $f_S \colon \zeta_\Sigma(S) \to \mathbb{R}$, evaluates the plan optimality when implemented in the system $S$. Therefore, we state the planning problem as an optimization problem as follows.

 \begin{equation}\label{eq:planprob_ex}
 \minimize_{\text{ s.t. } \forall i\  \hat{Y{_i}} \subset \Ima (\xi_\zeta)} f_S(\zeta),
 \end{equation}
 
To solve this problem optimally, the best trajectory in a continuous multi-dimensional state space must be found. However, real-time planning scenarios make this search infeasible. Therefore, we define sampled transition systems for each solver space to efficiently use Algorithm \ref{alg:planning} for plan generation and Algorithm \ref{alg:execution} for execution.

 \subsection{Sampled Transition Systems and Specifications for Monitoring Tasks}

Each transition system $S=(X,T,L,Y,H)$ defined from a multi-vehicle system of $J$ vehicles following Definition \ref{def:fulsys}, has a continuous state space. We define sampled transition systems $S_d=(X_d,T_d,L,Y_d,H)$ following the Definition \ref{def:sampled}. From the Monotonicity Theorem \ref{theo:mono}, we have the following result on sampled transition systems. 

\begin{corollary}\label{cor:sampled_mono}
    For a transition system $S=(X,T,L,Y,H)$ defined from a multi-vehicle system of $J$ vehicles following Definition \ref{def:fulsys}, a sampled transition system $S_d=(X_d,T_d,L,Y_d,H)$ of $S$ is monotone.
\end{corollary}

\begin{proof}
    From Corollary \ref{Cor:monotone}, we have that a transition system $S=(X,T,L,Y,H)$ defined from a multi-vehicle system of $J$ vehicles following Definition \ref{def:fulsys} is monotone. Further, from Theorem \ref{theo:mono}, we have that any transition system sampled from a monotone transition system is monotone. Therefore, a sampled transition system $S_d=(X_d,T_d,L,Y_d,H)$ is monotone. \end{proof}

Given a task site assignment $\Xi$ and a sampled transition system $S_d=(X_d,T_d,L,Y_d,H)$ of $S$, a specification $M \subset \Sigma(S_d)$ can be defined following Definition \ref{def:Specifcation}. For the defined specification $M$, we state the following result on the monotonicity of $M$. 

\begin{corollary}\label{cor:spec_mono}
    Given a transition system $S=(X,T,L,Y,H)$ defined from a multi-vehicle system of $J$ vehicles following Definition \ref{def:fulsys} and a task site assignment $\Xi$, the specification $M \subset \Sigma(S)$ defined from the task site assignment $\Xi$, is monotone.
\end{corollary}

\begin{proof}
    Let $\tau =x_{0}\delta_\gamma^{0}x_{1}\delta_\gamma^{1}\cdots\delta_\gamma^{N-1}x_{N} \in \Sigma(S)$ denote a trajectory in $M$ whose implementation is  $\zeta$. Let $\xi_\zeta$ be the output behavior of $\zeta$. Since, $\tau \in M$,  $\xi_\zeta$ satisfies $\Xi$. Let any trajectory $\tau^{\prime}
=x_{0}^{\prime}\ell_{0}x_{1}^{\prime}\ell_{1}\cdots\ell_{N-1}x_{N}^{\prime}$, where
\[
x_{0}^{\prime}\succeq x_{0},x_{1}^{\prime}\succeq x_{1}\cdots,x_{N}^{\prime
}\succeq x_{N}.% 
\]
That is, each $x_{i}^{\prime}$, which has higher energy and earlier time, ranks better than $x_{i}$.
%\JDH{It is unclear from this sentence whether $x_{i}^{\prime}$ or $x_i$ has the higher energy and lower time}
 From the definition of $\succeq$ we have $H(x_{i}^{\prime})=H(x_{i})$ for each $i$. Therefore, output behavior $\xi^\prime$ of an implementation of $\tau^\prime$ satisfies the task site assignment $\Xi$. Thus, $\tau^\prime$ is in $M$. Hence, the specification $M$ is monotone.
\end{proof}

From Corollary \ref{cor:sampled_mono} and Corollary \ref{cor:spec_mono}, we state that improvements for states are possible in the sampled transition systems defined from multi-vehicle systems, and when such improvements yield improved trajectories, they will still satisfy the intended specification. Therefore, we have established the possibility of improvement for a specification-satisfying plan.

Using the sampled transition systems $S_d$ and the defined specification $M\in \Sigma(S_d)$, we can restate the optimization problem in Equation \ref{eq:planprob_ex} as follows.

 \begin{equation}\label{eq:planprob_ex_sampled}
 \minimize_{\tau \in M} f_S(\zeta_\tau),
 \end{equation}

 From Theorem \ref{theo:task_sat}, we have that any $\tau \in M$, which is a solution to the optimization problem in Equation \ref{eq:planprob_ex_sampled}, yields an output behavior that satisfies the task site assignment from which $M$ was defined. Further, if the sampling for the creation of $S_d$ was conducted such that recursive feasibility conditions of Theorem \ref{theorem:feasibility} are satisfied, it ensures a specification-satisfying trajectory always exists in the sampled system.
 
 Now, we have completely defined a planning problem in a multi-vehicle system using the mathematical concepts of transition systems. We have also proved the necessary conditions for using Algorithm \ref{alg:planning} to iteratively employ a set of solvers to find a plan. Next, we present formulation of the solvers used in Algorithm \ref{alg:planning} for the monitoring task.

\subsection{Solver Formulation}

% we  use the running example to describe the proposed planning framework implementation. We have defined the original transition system and the planning problem for the running example in Section \ref{sec:origsys}.

In this section, we introduce the solvers used in the iterative planning framework and the respective sampled transition systems for these optimization solvers as follows for a given transition system for a multi-vehicle system $S$ and a monitoring tasks site assignment $\Xi$. 
\begin{enumerate}
    \item Team-level discrete optimization
    
    This solver operates in a sampled transition system $S_{d_1}$ of the original system $S$ and provides a feasible plan $\tau_{d_1} \in M_1$ for the whole team that satisfies the specification $M_1$ defined from $\Xi$.
    
    \item Agent-level optimization
    
    This solver operates in another sampled transition system $S_{d_2}$ and attempts to find a solution $\tau_{d_2}\in M_2$ such that it is better than the previous solution, i.e., $f_S(\zeta_{d_1})\geq f_S(\zeta_{d_2}) $ where $\zeta_{d_1},\zeta_{d_2}$ are the respective implementations of $\tau_{d_1}$ and $\tau_{d_2}$. Here, $M_2$ is the specification in $S_{d_2}$ defined from $\Xi$. This instance of the solver operates at each agent level and optimizes the motion trajectory assigned by the team-level plan. Therefore, it operates in the same discrete spacial points. However, the energy and time dimensions are continuous.
    
\end{enumerate}

Next, we look at the formulation of the previously described optimization solvers and their operating spaces.

\subsection{Team-level Discrete Optimization (TDO)}

The goal of the TDO is to generate a feasible plan that satisfies the specification for the whole team using minimal computation resources and time. Therefore, we employ a sampling procedure to discretize the state space and time to formulate sampled transition system.

\begin{figure}[h]
\centering
\includegraphics[width=0.35\textwidth]{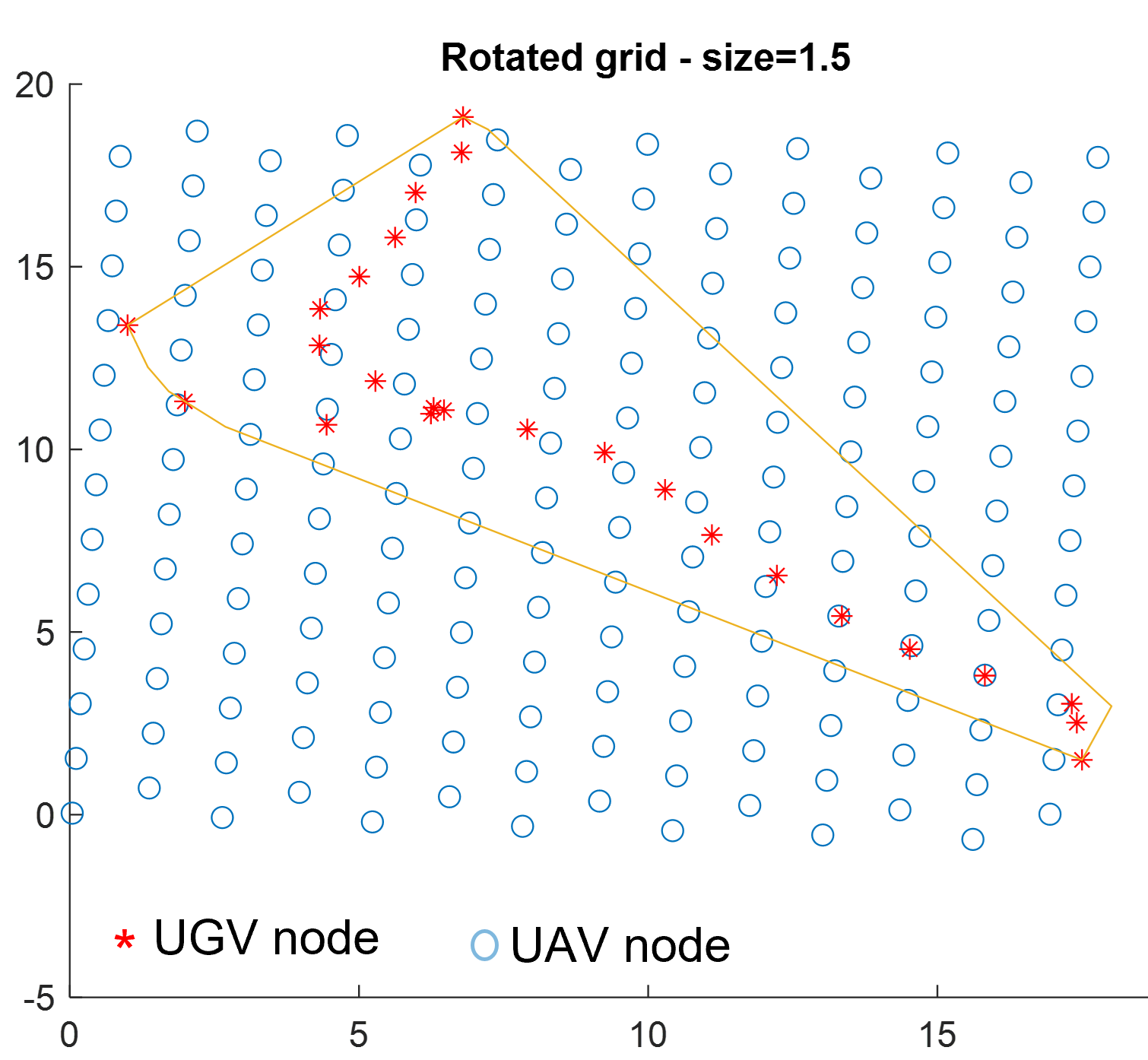}
\caption{Discretized example map. The region of operation is the spatial coordinates inside the convex hull.}
\label{fig:discretemap}
\end{figure}

\subsubsection{Time discretization}

The time discretization step $\gamma_d$ is chosen to set the resolution of the planning horizon. This can be chosen arbitrarily to match the specifications of the problem. In the running example, we set $\gamma_d=5\ min$.

\subsubsection{Spatial discretization}

The spatial dimension consists of all possible points the vehicles can travel in the known map $G$. However, the UGVs are restricted to the road denoted as $G_R$. As the time step $\gamma_d=5\ min$, we fix the velocity of the vehicles and evaluate the distance traveled by each vehicle in duration $\gamma_d$. For the UGVs, we discretize the road to a set of points $P_g=\{p_g |\ (p^x_g,p_g^{y})\in G_R \}$, which are $1.2\ km$ apart. The UAVs are capable of traveling between any points within the map $G$. Thus, we overlay a triangular grid of size $1.5km$ on the map to sample a set of points in the 2D space. We orient the grid such that maximum overlap is achieved with $P_g$ to remove redundant points. Then, we crop the grid to include the points within the convex hull of the region of operation to further reduce the state space size to set  $P_a=\{p_a |\ (p^x_a,p_a^{y}) \in G \}\cup P_g$. Here, $P_g \subset P_a$ allows UAVs to rendezvous with UGVs. The final spatial discretization is shown in Figure. \ref{fig:discretemap}. 

\subsubsection{Energy dimension discretization}
The transition of vehicles within the spatial nodes depletes their energy following the consumption curves in Eq. \eqref{eq:gndpower} and \eqref{eq:airpower}. The recharging of the UAV batteries either on a depot or on a ground vehicle replenishes the UAV batteries following the curve in Eq. \eqref{eq:aircharge}. Further, a UGV battery depletes following the same curve when recharging UAVs. However, from the time discretization, it is known that each depletion and recharge task is continued for a $\gamma_d$ interval per step. We calculate energy depletion per moving step for each vehicle assuming constant velocity for a $\gamma_d$ duration as $Pe_a(v_a)\times \gamma_d $ for UAVs and $Pe_g(v_g)\times \gamma_d $ for UGVs. Let the number of maximum energy levels of UAV be $B^{\mathrm{max}}_a$ and that of UGV be $B^{\mathrm{max}}_g$. We find the number of energy levels depleted as, $$B_{\mathrm{move}}^a=\lceil{\frac{B^{\mathrm{UAV}}}{B^{\mathrm{max}}_a}\times Pe_a(v_a)\gamma_d}\rceil$$ for the UAV, and $$B_{\mathrm{move}}^g=\lceil{\frac{B^{\mathrm{UGV}}}{B^{\mathrm{max}}_g}\times Pe_g(v_g)\gamma_d}\rceil$$ for the UGV. 
Then, we find the energy levels gained per charging step of $\gamma_d$ duration for the UAVs as $$B_{\mathrm{charge}}^a=\lfloor{\frac{B^{\mathrm{UAV}}}{B^{\mathrm{max}}_a}\times Pe_\mathrm{chg} \gamma_d}\rfloor.$$ At each charging step on the ground vehicle, $B_{\mathrm{charge}}^a$ is depleted from the UGV energy. We always over-approximate energy depletion through \textit{ceil} operation and under-approximate the energy gain by \textit{floor} operation.

Given the above discrete formulation of the state space, we define the optimization problem as an SMT problem to search for a feasible solution. This enables computing a feasible solution faster, as searching for the optimal solution is NP-hard. To this end, we define the set of constraints for the SMT problem that encodes the vehicle motion dynamics, energy dynamics, and specification satisfaction. We look at the formulation of such a constraint set for a team of one UGV and two UAVs.

Let a $K$ be the maximum length of the allowed time horizon. For each $k=1,\dotsc,K$ we define the following variables to formulate the SMT problem.

\begin{itemize}
    \item[1.] $p^i_a({k})\in P_a$- Position of UAV $i$ at time step $k$.
    \item[2.] $p_g({k})\in P_g$- Position of UGV at time step $k$.
    \item[3.] $b^i_a({k})\in \{0,\dotsc,B_a^{\mathrm{max}}\}$- Battery level of UAV $i$ at time step $k$.
    \item[4.] $b_g({k}) \in \{0,\dotsc,B_g^{\mathrm{max}} \}$- Battery level of UGV at time step $k$.
    \item[5.] $v^i_a({k})\in \{True,False\}$- Boolean variable denoting status of motion of UAV $i$ at time step $k$.
    \item[6.] $v_g({k})\in \{True,False\}$- Boolean variable denoting status of motion of UGV at time step $k$.
    \item[7.] $s^i(k) \in \{True,False\}$- Boolean variable denoting charging status of UAV $i$ on the UGV at time step $k$.
\end{itemize}

Using the defined variables, we state the following constraints on the motion of the vehicles. For the ground vehicle, 

\begin{equation}\label{eq:cons10}
\begin{split}
     \forall p,\  \forall k\;\ \   (p_g(k)=p) \land  &v_G(k) \\&\implies  \bigvee_{q \in P_{g}^{\mathrm{next}}(p)} p_g(k+1)=q.\\
\end{split}
\end{equation}
Here, $P_{g}^{\mathrm{next}}(p) \subset P_g$ denotes the next possible positions for a UGV currently at $p$. Similarly, we state the dynamical constraint for the $i^{th}$ UAV as, 

\begin{equation}
\begin{split}
     \forall p,\  \forall k\;\ \   (p^i_a(k)=p) &\land v^i_a(k) \land   \lnot(s^i(k)  \land s^i(k+1))\\& \implies \bigvee_{q \in P_{a}^{\mathrm{next}}(p^i)}  p^i_a(k+1)=q .
\end{split}
\end{equation}
Here, $P_{a}^{\mathrm{next}}(p) \subset P_a$ denotes the next possible positions for a UAV $i$ currently at $p$. Further, the UAV motion is described by the following constraints accounting for charging on UGV.

\begin{equation}
\begin{split}
 \lnot v^i_a(k) \land \lnot (s^i(k) \land s^i(k+1&))\\ &\implies p^i_a ({k+1 })=p^i_a ({k}),\\
     s^i(k) &\implies p^i_a(k)=p_g (k).
\end{split}
\end{equation}

Then, we state the constraints on the energy dynamics. UGV energy at time step $k$, depends on its velocity $v_g(k)$ and the status of UAV charging on UGV at $k$. We state the UGV energy dynamics as follows for the two UAV systems in focus.
\begin{equation}
\begin{split}
 \lnot v_g(k) \ \land & (s^1(k) \land s^2(k))\implies\\  & b_g ({k+1 })=b_g (k) -2 B^a_{\mathrm{charge}},\\
 \end{split}
\end{equation}
 \begin{equation}
\begin{split}
  \lnot v_g(k) \ \land & (s^1(k) \veebar s^2(k))\implies \\  &b_g ({k+1 })=b_g (k) -B^a_{\mathrm{charge}},\\
   \end{split}
\end{equation}
 \begin{equation}
\begin{split}
   \lnot v_g(k) \ \land & (\lnot s^1(k) \land \lnot s^2(k)) \implies  \\ & b_g ({k+1 })=b_g (k),\\
    \end{split}
\end{equation}
 \begin{equation}
\begin{split}
   v_g(k) \ \land &(s^1(k) \land s^2(k))  \implies\\ &  b_g  ({k+1 })=b_g  (k) -2 B^a_{\mathrm{charge}}-B^g_{\mathrm{move}},\\
    \end{split}
\end{equation}
 \begin{equation}
\begin{split}
  v_g(k) \ \land &(s^1(k) \veebar s^2(k))  \implies \\ &  b_g  ({k+1 })=b_g  (k) -B^a_{\mathrm{charge}}-B^g_{\mathrm{move}},\\
   \end{split}
\end{equation}
 \begin{equation}
\begin{split}
   v_g(k) \ \land &(\lnot s^1(k) \land \lnot s^2(k))  \implies\\ &  b_g  ({k+1 })=b_g  (k)-B^g_{\mathrm{move}}.
\end{split}
\end{equation}

Similarly, we state the constraints for energy dynamics for the UAV $i$ as follows. Let $Dp^i$ be a Boolean variable denoting whether the UAV $i$ is at a depot location. Here, $p_{\mathrm{depot_1}},p_{\mathrm{depot_2}},p_{\mathrm{depot_3}}$ denotes the discrete nodes corresponding to depot locations.
% \begin{equation}
% \begin{split}
%      Dp^i=(p^i_a(k)=p_{\mathrm{depot_1}} &\lor p^i_a(k)=p_{\mathrm{depot_2}}& \lor p^i_a(k)=p_{\mathrm{depot_3}}).\\
% \end{split}
% \end{equation}

\begin{equation}
\begin{split}
     Dp^i&=(p^i_a(k)=p_{\mathrm{depot_1}} \lor p^i_a(k)=p_{\mathrm{depot_2}}\\ & \lor p^i_a(k)=p_{\mathrm{depot_3}}).
\end{split}
\end{equation}

The following constraints state the energy dynamics of the UAV when moving, charging on a UGV, or charging at a depot location.
 \begin{equation}
\begin{split}
 \lnot s^i(k) \land \lnot v^i_a(k) \land \lnot Dp^i &\implies b^i_a ({k+1 })=b^i_a (k),\\
    \end{split}
\end{equation}
 \begin{equation}
\begin{split}
  \lnot s^i(k) \land v^i_a(k) \land &\lnot Dp^i \implies\\& b^i_a ({k+1 })=b^i_a (k)-B^a_{\mathrm{move}},\\
     \end{split}
\end{equation}
 \begin{equation}
\begin{split}
  ( s^i(k) \lor Dp^i)\land \lnot &v^i_a(k)  \implies \\&b^i_a ({k+1 })=b^i_a (k)+B^a_{\mathrm{charge}},\\
     \end{split}
\end{equation}
 \begin{equation}
\begin{split}
  ( s^i(k) \lor Dp^i)\land &v^i_a(k)  \implies \\&b^i_a ({k+1 })=b^i_a(k)+B^a_{\mathrm{charge}} - B^a_{\mathrm{move}}.
   \end{split}
\end{equation}

From the discretization of the spatial dimension, the original task of visiting $G_R$ reduces to a specification of visiting all discrete positions in $P_g$. Thus, we state our monitoring task as a constraint as follows.

\begin{equation}\label{eq:cons24}
    \forall p \in P_g \;\ \  \bigvee_{k=0}^K \big( p_g(k)=p \lor \bigvee_{\forall i} p^i_{a}(k)=p \big) =True.
\end{equation}

This constraint is satisfied for each location in $P_g$, when it is visited by any vehicle at any time step before $K$. Therefore, its satisfaction ensures the satisfaction of the task site assignment. Now, any solution $\tau$ that satisfies the previously defined set of constraints is a feasible solution to the Team-level discrete optimization problem. We employ $Z3$, an SMT solver, to search for a feasible solution faster.

The feasible solution that we gain after TDO is sub-optimal due to three kinds of sub-optimalities introduced as follows. 

\begin{itemize}
    \item[1.] Sub-optimality of the SMT solver. 
    
    The SMT solver only searches for constraint satisfaction. Therefore, the produced solution is not optimal.
     \item[2.] Sub-optimality due to discretization. 
     
    The state space is a discrete space generated through the sampling process. Each state transition takes place in an interval $\gamma_d$. Therefore, the solution trajectories could be sub-optimal when compared to continuous time, and continuous state trajectories.
    \item[3.] Sub-optimality due to conservativeness of the sampled system. 
    
    The sampled system is generated by assuming a constant velocity of the vehicles. However, when a generated plan is executed in the original infinite state system, a suitable optimal velocity can be selected by the lower-level control algorithm. Further, due to the over-approximations of the energy depletion and under-approximation of energy gained while charging, the sampled system is more conservative than the original infinite state system. This introduces further sub-optimality to the TDO solution.
   
\end{itemize}

TDO is used to generate a feasible solution quickly. Thus, the remainder of allocated computation resources and time can be exploited toward reducing the sub-optimality of the solution. To this end, we use agent-level optimization as the next solver in the iterative planning framework, where we further optimize the solution generated by TDO. 

\subsection{Agent-level  Optimization (AO)}

We identify that the TDO solution $\tau_{d_1}$ can be optimized at each agent level in between events involving multiple agents. From the TDO solution, we make the following observations. 

\begin{itemize}
    \item[1.] The traversal order of the visits to task site coordinates can be modified.
    \item[2.] All charging events can be dropped or reordered, as long as the energy dynamic constraints are not violated.
    \item[3.] Charging on the UGV is equivalent to charging at depots. However, they are only active at a time step where the UGV is at that specific spatial coordinate.
\end{itemize}

Based on the previous observations, improvements are possible to the AO solution $\tau_{d_2}$, such that  $f_S(\zeta_{d_1})\geq f_S(\zeta_{d_2})$, where $\zeta_{d_1},\zeta_{d_2}$ are respective implementations of $\tau_{d_1}$ and $\tau_{d_2}$. For each UAV, we formulate a vehicle routing optimization problem as an SMT problem. Therefore, we operate within a set of discrete coordinates in the spatial dimension while allowing the time and energy dimensions to be continuous. To guarantee the recursive feasibility of $\zeta_{d_1}$ at AO, the conditions of Theorem \ref{theorem:feasibility} must be satisfied. For some set of time stamps $\Gamma=\{\gamma_0,\dotsc,\gamma_\mu\}$, let $\Tilde{X}= \{x\ | \  x=\zeta_{d_1}(\gamma)\ \exists \gamma\in \Gamma \} $. Since time and energy dimensions are continuous, we only need to ensure that UAV can reach $\Tilde{X}$ projected to the spatial dimension. Let $\Tilde{Y}= \{y\ | \ \exists x\in \Tilde{X}\  y \in H(x) \} $.   We define the following variables for each UAV.

\begin{itemize}
    \item[1.] $N^v=\{N^v_1,\dotsc,N^v_n\}$- The set of coordinates to be visited to satisfy the goal.
    \item[2.] $N^d=\{N^{d}_1,\dotsc,N^{d}_r\}$- The set of coordinates of depot locations where charging is possible.
    \item[3.] $N^g\{=\{N^{g}_1,\dotsc,N^{g}_q\}$- The set of coordinates of locations where charging is possible on the UGV.
    \item[4.] $\gamma^{g}_1,\dotsc,\gamma^{g}_q$- The set of timestamps where each UGV charging location is active.
    \item[5.] $\rho(k)\in N^v \cup N^d \cup N^g \cup \Tilde{Y} ,k=0,\dotsc,K$- The position of UAV at each transition $k$.
     \item[6.] $\gamma(k) ,k=0,\dotsc,K$- The time stamp at each transition $k$.
    \item[7.] $\eta(k),k=0,\dotsc,K$- The energy of UAV at each transition $k$.
\end{itemize}

At each transition from a position $\rho_i$ to $\rho_j$, the following constraint must be satisfied.

\begin{equation}
\begin{split}
    \rho(k)&=\rho_i \land \rho(k+1)=\rho_j\\ &\implies(\gamma({k+1})=\gamma(k)+\gamma_{ij}+\gamma_j),\\
        \rho(k)&=\rho_i \land \rho(k+1)=\rho_j\\ &\implies (\eta({k+1})=\eta(k)+\eta_{ij}+\eta_j),
    \end{split}
\end{equation}
where $\gamma_{ij}$ is the transition time from $\rho_i$ to $\rho_j$, $\gamma_j$ is the slack time at $\rho_j$ position, $\eta_{ij}$ is the transition energy from $\rho_i$ to $\rho_j$ and $\eta_j$ is the slack energy at $\rho_j$ position. We evaluate $\gamma_{ij}$ as follows.

\begin{equation}
    \gamma_{ij}=
\begin{cases}
\gamma_d, \ \ i=N^{g}_{q_i} \land  j=N^{g}_{q_j} ,\\
\phi_\gamma(i,j), \ \ else.\\
\end{cases}
\end{equation}
Here, $\phi_\gamma(i,j)$ is a function that evaluates the transition time for any given pair of nodes when traveling at UAV velocity. The slack time, $\gamma_j$, is set as follows.

\begin{equation}
    \gamma_j=
\begin{cases}
\gamma_d, \ \ j=N^{d}_q,\\
0, \ \ else.\\
\end{cases}
\end{equation}

Then, we evaluate the transition energy $\eta_{ij}$ as follows.
\begin{equation}
    \eta_{ij}=
\begin{cases}
B^a_{\mathrm{charge}}, \ \  i=N^{g}_{q_i} \land  j=N^{g}_{q_j},\\
\psi_\eta(i,j), \ \ else.\\
\end{cases}
\end{equation}

Here, $\psi_\eta(i,j)$ is a function that evaluates the transition energy consumed between any given pair of nodes when traveling at UAV velocity. The slack energy, $\eta_j$ is set as follows.

\begin{equation}
    \eta_j=
\begin{cases}
B^a_{\mathrm{charge}}, \ j=N^{d}_{r},\\
0,\ \ else.\\
\end{cases}
\end{equation}

The following constraint is added to ensure  that the  specification is satisfied.
\begin{equation}
    \bigwedge_{\forall N^v}\Big(\bigvee_{\forall \rho(k)}(\rho(k)=N^v)\Big)=True
\end{equation}

Further, any UGV charging location is set to be active at the corresponding time as follows.
\begin{equation}
    \rho(k)=N^{g}_q \implies \gamma(k)=\gamma^{g}_q.
\end{equation}

At each depot charging instance, we set the start time to be an integer multiple of $\gamma_d$.

\begin{equation}
    \bigwedge_{\forall N^{d}_r}\Big(\bigvee_{\forall i\in\{1,...,K\}}(\rho(k)=N^{d}_r \implies \gamma(k)= i\gamma_d)\Big)=True.
\end{equation}

The AO problem is formulated using the previously defined constraints for each vehicle. Any feasible solution $\tau$ that will satisfy the set of constraints satisfies the mission. An SMT solver, $Z3$, is used to generate a feasible solution. However, the solution improvement is imposed as follows in Algorithm \ref{alg:planning}.
\begin{equation}
     f_S(\zeta_{d_1})\geq f_S(\zeta_{d_2}).
\end{equation}

 AO solution improves the plan from TDO for each UAV agent while satisfying the task site assignments. Above formulation allows both TDO and AO to be used as solvers in the proposed iterative framework in Algorithm \ref{alg:planning}. With sampled systems defined, we set $\Omega_1=$ TDO and $\Omega_2=$ AO. This is because TDO can compute feasible solutions quickly, while AO can reduce the  sub-optimality while exploiting computation resources. Further, TDO and AO are implemented in the shrinking horizon execution scheme described in Algorithm \ref{alg:execution}, which improves the plan further in between execution steps. In the next section, we analyze the results of this implementation of the proposed solution framework for planning tasks in the running example.

\section{Implementation Results}

\subsection{Direct Implementation vs. Proposed Implementation}
This section presents the results of implementing Algorithm \ref{alg:planning} and Algorithm \ref{alg:execution} for a multi-vehicle planning scenario, using the two aforementioned solvers, TDO and AO. The task is to compute a plan for the road monitoring task described in the previous section. We use a team of two UAVs and one UGV. The discretization time $\gamma_d$ is chosen as 5 min. Therefore, the shrinking horizon execution step is 5-min long. We also set the timeouts of the solver such that the execution step is not violated. The planning algorithm was executed on an AMD Ryzen Threadripper processor at 3.9MHz with 64GB RAM. We compare the results of the following two  approaches. 
\begin{itemize}
    \item[1.] Open-loop execution of TDO (Direct Implementation).

    A plan is computed using TDO and implemented for the whole planning horizon. Therefore, we do not employ a secondary solver for reducing sub-optimality nor trigger any re-planning after the first computation.
    
\item[2.] Shrinking horizon execution of iterative planning (Proposed Implementation).

A plan is computed using Algorithm \ref{alg:planning} employing TDO and AO iteratively at each step of execution through a shrinking horizon using algorithm \ref{alg:execution}. Therefore, two solvers are used iteratively in planning, and at each step of execution, re-planning is triggered. 
    
\end{itemize}

\begin{figure}[h]
\centering
\includegraphics[width=0.5\textwidth]{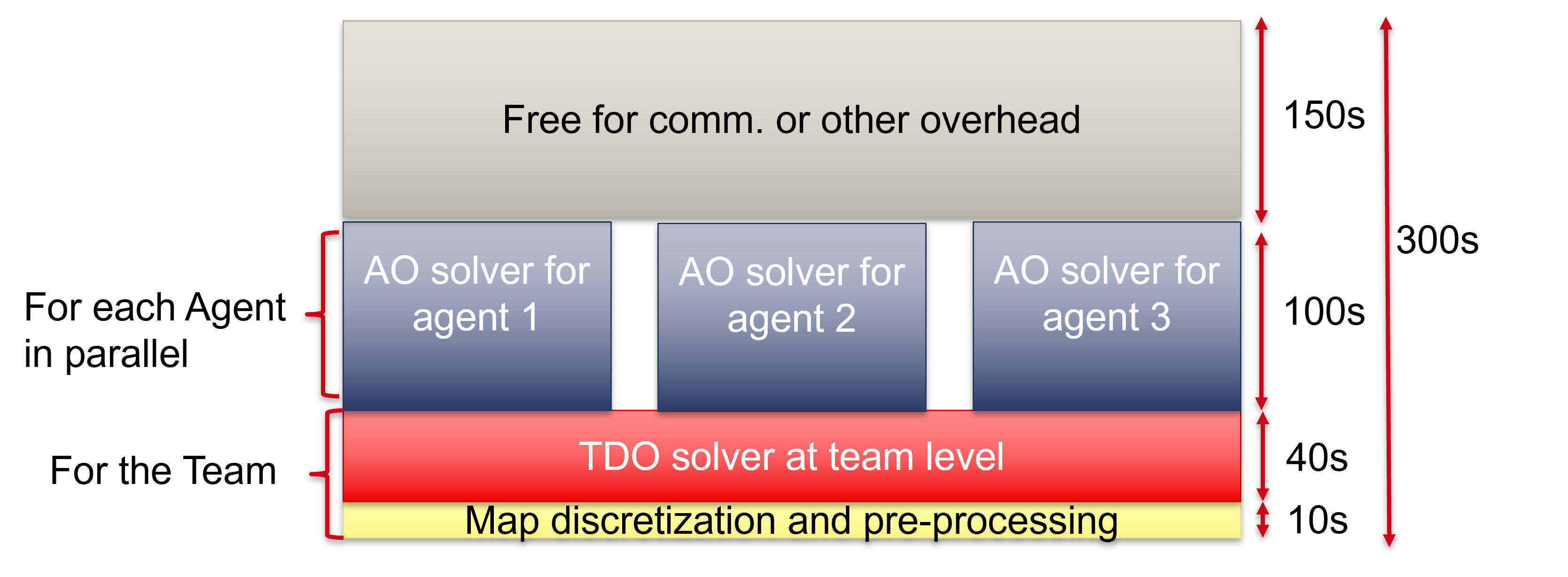}
\caption{Average compute time distribution: TDO represents Direct Implementation, while TDO+AO represents the proposed approach.}
\label{fig:compute_time}
\end{figure}

Figure \ref{fig:compute_time} shows the distribution of compute time for the proposed approach. We run TDO to compute a plan for the whole team for the complete horizon. If this plan is implemented, it corresponds to the Direct Implementation approach. However, using the proposed approach, we run Algorithm \ref{alg:planning} with AO as the next step. As it is run at each agent individually, this enables running AO in parallel. For an execution step of $5$ min, we finish planning through the proposed approach leaving a substantial amount of time for communications and other practical overheads. From the results of running the proposed approach, it requires an extra $100$ seconds of computing time on average than direct implementation, while still staying within the boundary of the 5-min execution step. However, by trading off execution time, the proposed approach gains an improvement of the total mission time.

\begin{table}[h!]
    \caption{Result comparison between Direct Implementation and Proposed Implementation }
\begin{center}
\begin{tabular}{ | m{2cm} | m{2cm}| m{2cm} | } 
  \hline
  Comparison metric & Direct Implementation & Proposed Implementation \\ 
  \hline
  Plan time & 120mins & 105mins \\ 
  \hline
  Compute time & 50s & 150s\\ 
  \hline
  %   Plan time improvement  &  - &  12\% \\ 
  % \hline
\end{tabular}
\end{center}
\label{table:result}
\end{table}

We present the results of the mission time improvement and compute time in Table \ref{table:result}. It is evident from the results that the proposed algorithm minimizes plan time and improves the cost on an average of $12\%$ by trading off compute time up to 100s. Further, Figure \ref{fig:improvem} shows the average improvement of planning time $\delta_f$ across each execution step for 25 instances of planning solutions generated by the proposed approach. If $\zeta_{TDO}$ and $\zeta_{AO}$ are implementations of plans generated from TDO and AO, respectively, we define, 
\begin{equation}
     \delta_f= \frac{f_S(\zeta_{TDO})-f_S(\zeta_{AO})}{f_S(\zeta_{TDO})}.
\end{equation}
 We observe that the iterative planning framework and the shrinking horizon execution method contribute to continual improvements in the total mission time at each time step of execution.

\begin{figure}[h]
\centering
\includegraphics[width=0.48\textwidth]{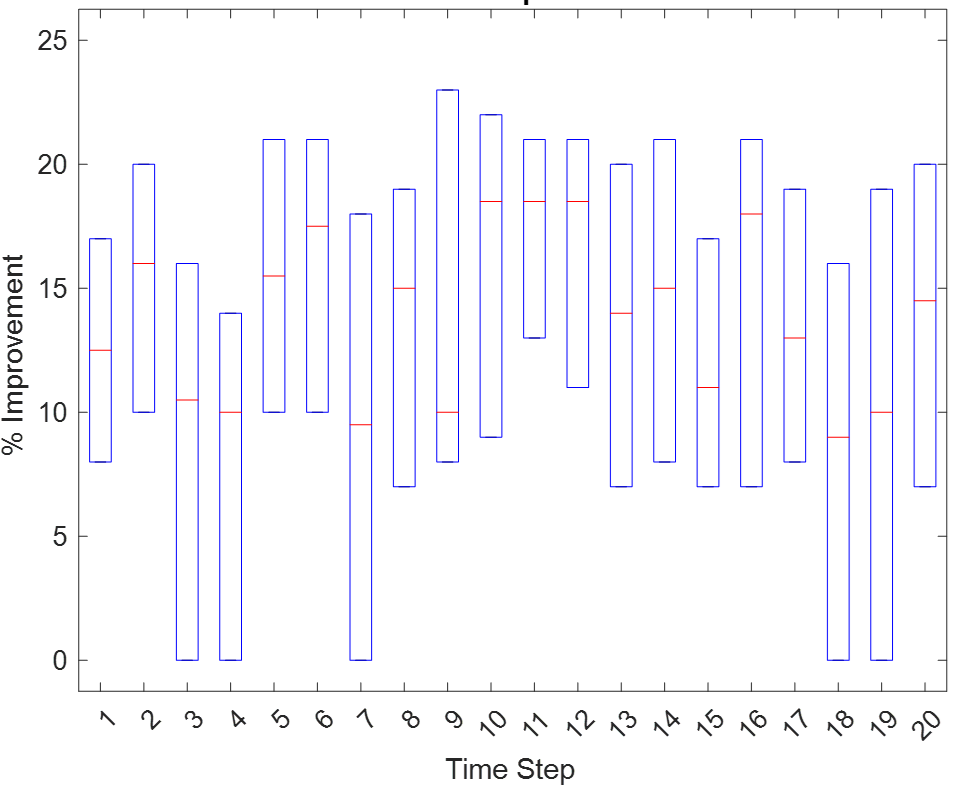}
\caption{Average planning time improvement $\delta_\gamma$: The plot shows the mean and the interquartile range of plan time improvement at each execution step, averaged for 25 planning instances.}
\label{fig:improvem}
\end{figure}

\color{black}
\subsection{Comparison: Proposed Implementation vs. Genetic Algorithm and Bayesian Optimization Solvers }

\color{black}
\subsubsection{Genetic Algorithm and Bayesian Optimization Solvers}

This section presents the results of the proposed implementation (PI) compared with a solver based on a genetic algorithm (GA) and a solver based on Bayesian optimization (BO) for the multi-vehicle planning scenario as presented in \cite{7_9836044}. We selected these specific baselines as \cite{7_9836044} uses the same energy-aware cooperative routing scenario described in Section \ref{sec:example} as the case study to evaluate the proposed solutions. Further, the approach taken in \cite{7_9836044} follows the state of the art, which converts the original continuous state space into a hybrid state space. It also uses a hierarchical approach, where an initial step of clustering of task sites simplifies the problem for the next optimization step. The utilized problem parameters are detailed in the table \ref{table:parame}. 

\begin{table}[h!]
    \caption{Problem parameters in compared case study }
\begin{center}
\begin{tabular}{ | m{1.5cm} | m{1.75cm}|  m{1.75cm}| m{1.75cm} | } 
  \hline
   Parameters & GA & BO & PI \\ 
  \hline
  UGV & 1 &  1 & 1 \\ 
  \hline
  UAVs & 2 & 2 & 2\\ 
  \hline
Scenario & Figure \ref{fig:exsample_scn} & Figure \ref{fig:exsample_scn} & Figure \ref{fig:exsample_scn}\\ 
  \hline
  Depots & 3 & 3 & 3\\ 
  \hline
    Discrete states & 41 & 41 & 38\\ 
  \hline
    Energy dynamics & Eqs. \eqref{eq:aircharge}, \eqref{eq:airpower}, \eqref{eq:gndpower} & Eqs. \eqref{eq:aircharge}, \eqref{eq:airpower}, \eqref{eq:gndpower}  & Eqs. \eqref{eq:aircharge}, \eqref{eq:airpower}, \eqref{eq:gndpower} \\ 
  \hline
  %   Plan time improvement  &  - &  12\% \\ 
  % \hline
\end{tabular}
\end{center}
\label{table:parame}
\end{table}
\subsubsection{Result Comparision}
The team composition was two UAVs and one UGV. The planning algorithms, PI and GA, were executed on an AMD Ryzen Threadripper processor at 3.9MHz with 64GB RAM while BO was executed on a server with similar computing power. We compare the results of the two approaches BO and GA against the proposed implementation. The detailed implementations of the GA-based solver and BO-based solver are discussed in \cite{7_9836044}. The tested scenario was to complete the road monitoring task for the road network in Figure \ref{fig:map}. The results of this implementation are presented in Table\ref{table:result_comp}. We observe that the proposed approach achieves faster compute time while achieving more than $50\%$ reduction of the total mission time. Due to the fast solution computation, it is evident that the proposed shrinking horizon execution scheme can be implemented in a real-time planning scenario.

\color{black}

\begin{table}[h!]
    \caption{Result comparison between Genetic Algorithm (GA), Bayesian Optimization (BO) and Proposed Implementation (PI) }
\begin{center}
\begin{tabular}{ | m{1.5cm} | m{1.75cm}|  m{1.75cm}| m{1.75cm} | } 
  \hline
  Metric & GA & BO & PI \\ 
  \hline
  Plan time & 225min &  249min & 105min \\ 
  \hline
  Compute time & 94min & 15min & 2.5min\\ 
  \hline
  %   Plan time improvement  &  - &  12\% \\ 
  % \hline
\end{tabular}
\end{center}
\label{table:result_comp}
\end{table}

\color{black}

\subsection{Comparison: Proposed Implementation vs. Optimization Solver }

In this section, we compare the case study results of the proposed approach against optimal solutions from a Mixed Integer Programming (MIP) solver. One of the key advantages of our approach compared to optimization solvers is that we have an anytime property while \color{black} optimization solvers, especially in real-life larger-scale NP-hard problems, often fail to deliver an optimal solution in a compute time that enables real-time implementation. \color{black} The following results further validate the strengths of the proposed approach.  
% \subsubsection{Optimization solver}
We employ an MIP solver (Google ORtools) that solves the optimization problem \eqref{eq:planprob_ex_sampled_case}, which minimizes the total plan times subject to the system and task constraints in  Eqs. \eqref{eq:cons10} - \eqref{eq:cons24}.

\begin{equation}\label{eq:planprob_ex_sampled_case}
 \minimize_{\tau \in M} f_S(\zeta_\tau),\\
 \text{  subject to : Eqs. \eqref{eq:cons10} - \eqref{eq:cons24} .}
 \end{equation}

The results of MIP solver, which is an optimal solver (OP) and Proposed Implementation (PI) is compared to evaluate the sensitivity to the number of agents and the number of task sites. We use the compute time and the sub-optimality as the evaluation metrics.

\subsubsection{ No. of Agents} We adjust the number of agents by setting the number of groups of agents. A group consists of a UGV and two UAVs. We adjust the number of groups from $2-12$, which changes the number of agents from $6-36$. The number of task sites is held lower at $5$. We set an upper bound of compute time as $T_d=300s$, as that is the execution time step. It is considered a failed attempt if MIP or PI fails to deliver a solution within $T_d$.

\begin{figure}[h]
\centering
\includegraphics[width=0.48\textwidth]{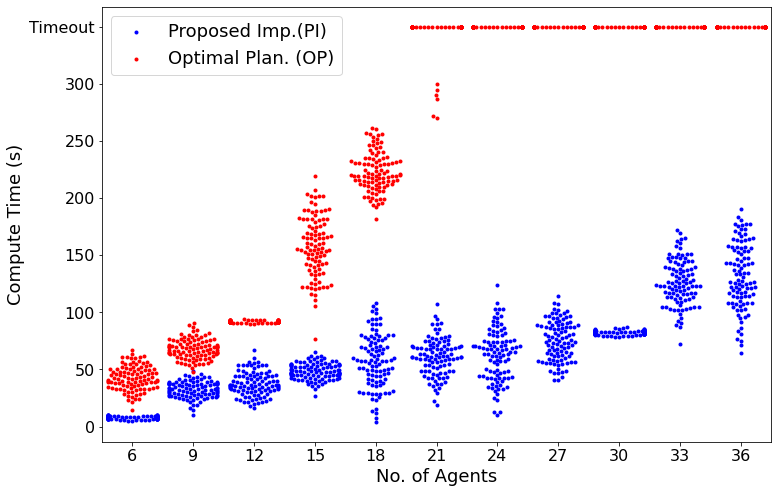}
\caption{\color{black}Swarm plot of compute time distribution OP and PI for varying the no. of agents. PI - Each run generates a feasible solution. OP - Each run generates the optimal solution. Timed out runs are shown at the top.}
\label{fig:comp_t}
\end{figure}

We present the results of this case study in Figure \ref{fig:comp_t}. We see that the proposed approach generates feasible solutions faster than the optimal solution, as expected. While the compute time increases with the no. of agents, it is still within $T_d$. Thus, a feasible task-satisfying solution can always be generated. This is because when the number of agents increases, the feasible solution space also increases in size. Thus, finding a feasible solution is easier and faster, while searching for the optimal solution is computationally intensive. 

\color{black} 
The plan time for completed runs is shown in Figure \ref{fig:plantime_agents}, which excludes the timed-out runs. 
The PI solutions were sub-optimal, taking a maximum of $22.2\%$ additional time than optimal solutions. The optimality gap decreases with increasing agents. This is because the agent-level optimization of iterative planning works best to improve the optimality gap when more variables are present at a higher number of agents.
\color{black}  

\begin{figure}[h]
\centering
\includegraphics[width=0.48\textwidth]{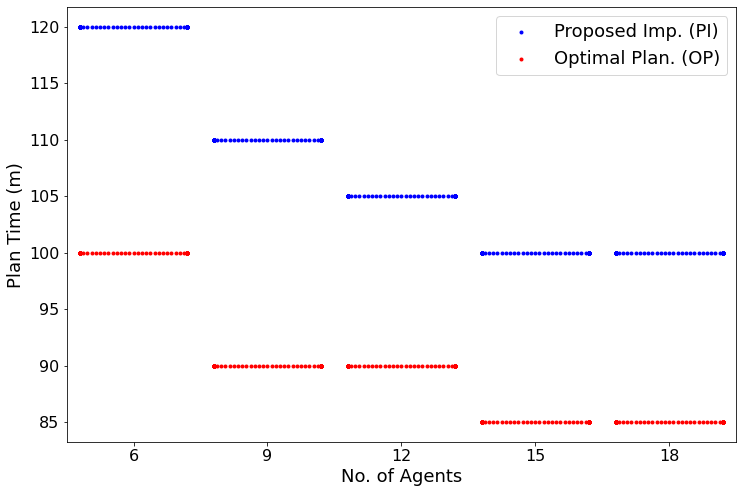}
\caption{Plan time distribution OP and PI for varying the no. of agents. PI - Each run generates a feasible solution. OP - Each run generates the optimal solution. \color{black} Maximum sub-optimality is 22.2\%.}
\label{fig:plantime_agents}
\end{figure}

\subsubsection{ No. of Task Sites} We adjust the number of task sites by adjusting the goal states for the monitoring task from $3-20$. The number of agents is held at $9$. We set an upper bound of compute time as $T_d=300s$, as that is the execution time step. It is considered a failed attempt if MIP or PI fails to deliver a solution within $T_d$.
\begin{figure}[h]
\centering
\includegraphics[width=0.48\textwidth]{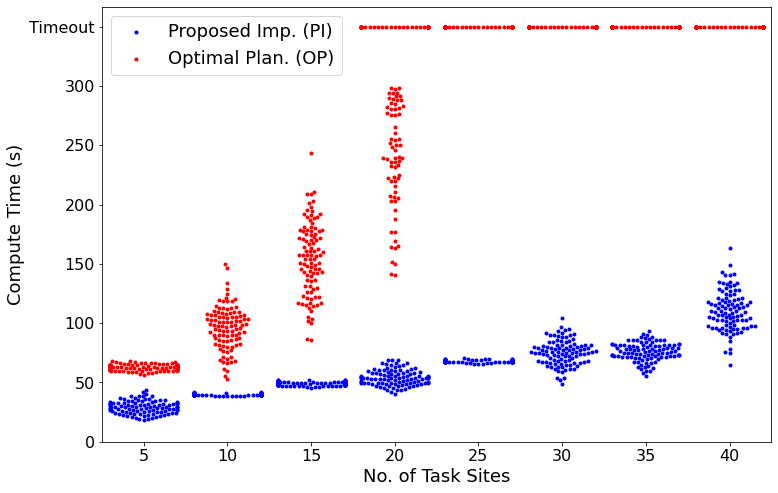}
\caption{\color{black} Swarm plot of compute time distribution OP and PI for varying the number of task sites. PI - Each run generates a feasible solution. OP - Each run generates the optimal solution. Timed out runs are shown at the top.}
\label{fig:comp_t_sites}
\end{figure}

We present the results of this case study in Figure \ref{fig:comp_t_sites}.
We see that the proposed approach generates feasible solutions faster than the optimal solution, as expected. While the compute time increases with the number of task sites, it is still within $T_d$. Thus, a feasible task-satisfying solution can always be generated. This is because when the number of task sites increases, the number of constraints that need to be satisfied increases. Therefore, searching for a satisfying solution becomes computationally harder, while searching for the optimal solution is further computationally expensive. The plan time for completed runs is shown in Figure \ref{fig:plantime_tasks}, \color{black}which excludes the timed-out runs. \color{black} 
On average, the solutions from PI were sub-optimal at $10.5\%$ additional time than optimal solutions.

\begin{figure}[h]
\centering
\includegraphics[width=0.48\textwidth]{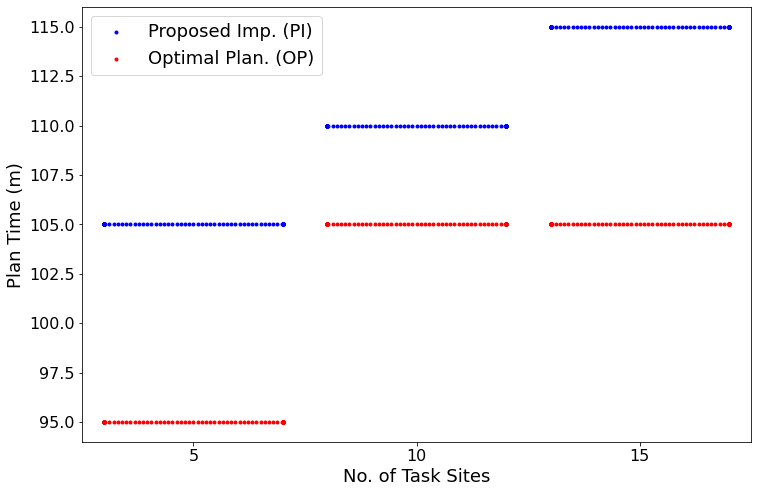}
\caption{\color{black}Plan time distribution OP and PI for varying the no. of task sites. PI - Each run generates a feasible solution. OP - Each run generates the optimal solution. \color{black}Maximum sub-optimality is 10.5\%.}
\label{fig:plantime_tasks}
\end{figure}

\color{black}

\color{black}

\section{Conclusion}
In this paper, we present an iterative planning framework for multi-agent systems. We formulate a mathematical framework to represent multi-dimensional systems and planning tasks using transition systems. Solving planning problems of this nature is computationally complex. Therefore, we propose the idea of using sampled transition systems to discretize the problem to employ faster solvers. We propose an iterative approach that defines a solution framework using multiple solvers to continually improve a plan using a given amount of computation time. The proposed approach ensures that a solution is always available for execution, while the iterative process continues on further optimization. We also propose a shrinking horizon execution of the plan to enable re-planning at each step of execution for continual improvement of the plan. Further, we also derive and prove the necessary conditions that guarantee the recursive feasibility of the proposed approach. The proposed theoretical framework has the ability to be applied to plan for any generalized task site assignment using multiple solvers iteratively. 

We validate the proposed approach using an energy-aware UAV-UGV cooperative route planning problem to satisfy a given task site assignment. The example presented in this paper validates the proposed framework by producing plans that satisfy the goal, using computing resources efficiently to preserve real-time implementation ability. The iterative planning framework consisting of two SMT solvers ensures the continual improvement of the plan when executed in a shrinking horizon. 

In conclusion, the theoretical framework introduced in this paper to represent complex dynamical systems presents several avenues for future improvements. The proposed framework can be further extended to define planning tasks that go beyond the generalized task site assignments. Further, exploring methods to integrate uncertainty present in plan execution and communication into the theoretical framework is important. We believe the proposed framework has the potential to serve as a useful tool for researchers and practitioners working on multi-agent system planning problems.

% \color{red}

% {\appendix[Proof of the Zonklar Equations]
% Use $\backslash${\tt{appendix}} if you have a single appendix:
% Do not use $\backslash${\tt{section}} anymore after $\backslash${\tt{appendix}}, only $\backslash${\tt{section*}}.
% If you have multiple appendixes use $\backslash${\tt{appendices}} then use $\backslash${\tt{section}} to start each appendix.
% You must declare a $\backslash${\tt{section}} before using any $\backslash${\tt{subsection}} or using $\backslash${\tt{label}} ($\backslash${\tt{appendices}} by itself
%  starts a section numbered zero.)}

% {\appendices
% \section*{Proof of the First Zonklar Equation}
% \section*{Proof of the Second Zonklar Equation}
% Appendix two text goes here.}

\bibliographystyle{IEEEtran}
\bibliography{bibdata}

\begin{IEEEbiography}[{\includegraphics[width=1in,height=1.25in,clip,keepaspectratio]{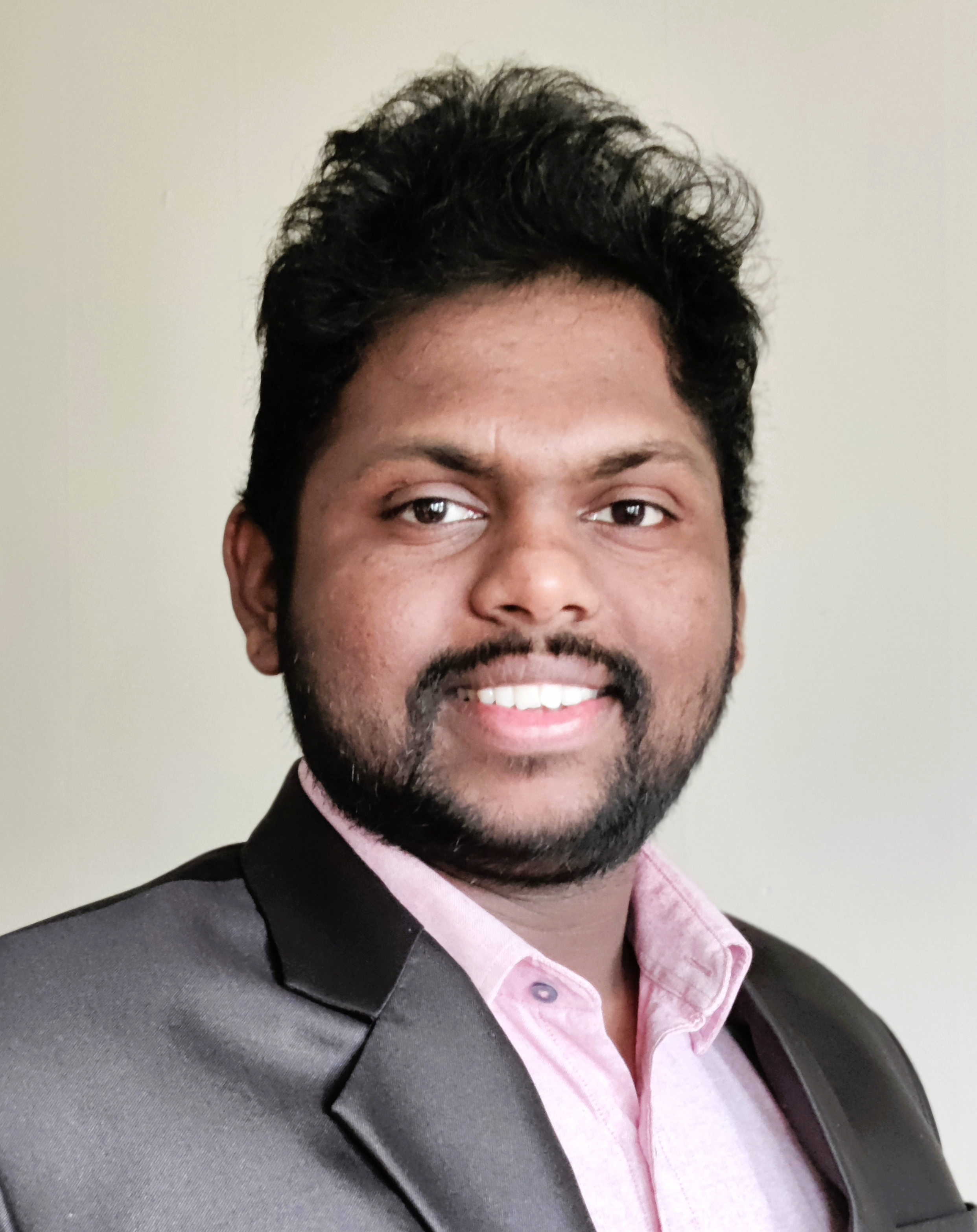}}]{Neelanga Thelasingha} is a Ph.D. candidate in the Electrical Computer and Systems Engineering Department at Rensselaer Polytechnic Institute, New York, USA. He received his M.S. in Computer and Systems Engineering from Rensselaer Polytechnic Institute, New York, USA, in 2021 and his B.Sc. in Electrical and Electronic Engineering with first-class honors from the University of Peradeniya, Sri Lanka, in 2017. His research interests include optimal control and motion planning in multi-agent autonomous systems.
\end{IEEEbiography}
\vspace{-33pt}
\begin{IEEEbiography}[{\includegraphics[width=1in,height=1.25in,clip,keepaspectratio]{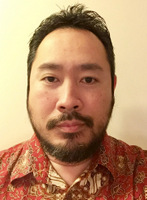}}]{A. Agung Julius (M’06, SM'21)} received the Ph.D. degree in applied mathematics from the University of Twente, Twente, The Netherlands, in 2005. From 2005 to 2008, he was a postdoctoral researcher with the University of Pennsylvania. Since 2008, he has been with the Department of Electrical, Computer, and Systems Engineering, Rensselaer Polytechnic Institute, Troy, New York, where he is currently a professor. His research interests include systems and control, systems biology, stochastic models in systems biology, control of biological systems, hybrid systems, and mathematical systems theory. He received the National Science Foundation CAREER Award in 2010. He is a senior member of the IEEE.
\end{IEEEbiography}
\vspace{-33pt}
\begin{IEEEbiography}[{\includegraphics[width=1in,height=1.25in,clip,keepaspectratio]{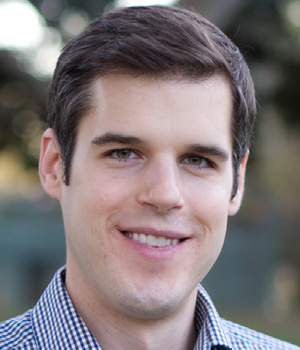}}]{James Humann} has been with the DEVCOM Army Research Laboratory since 2017 and is currently a Mechanical Engineer in the Vehicle Applied Research Branch. His research is focused on design, path planning, and simulation for multi-robot systems. He earned a Ph.D. and M.S. in Mechanical Engineering from the University of Southern California, and a B.S. in Mechanical Engineering from the University of Oklahoma.
\end{IEEEbiography}
\vspace{-33pt}
\begin{IEEEbiography}[{\includegraphics[width=1in,height=1.25in,clip,keepaspectratio]{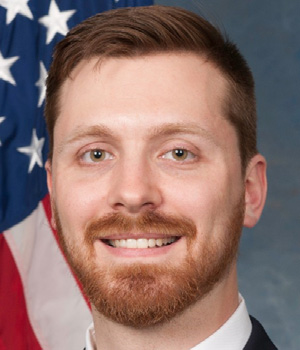}}]{Jean-Paul Reddinger} is an Aerospace Research Engineer at DEVCOM Army Research Laboratory. He currently works on Unmanned Aerial Systems (UAS) flight dynamics for modeling and simulation as part of the Vehicle Applied Research Branch. His background is in aeromechanics and controls for high-speed
compound rotorcraft.
\end{IEEEbiography}
\vspace{-33pt}
\begin{IEEEbiography}[{\includegraphics[width=1in,height=1.25in,clip,keepaspectratio]{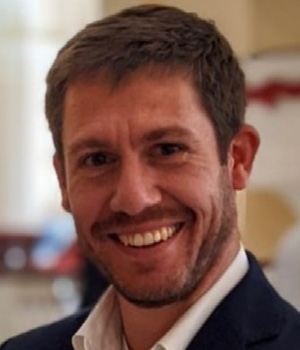}}]{James Dotterweich}
is a Robotics Research Engineer with the DEVCOM Army Research Laboratory in Aberdeen, Maryland, USA. He received his Master of Science in Mechanical Engineering with a robotics focus from The University of Utah, Salt Lake City, UT, USA in 2015. He received his Bachelor’s degree in Aerospace Engineering from The University of Colorado at Boulder, CO, USA in 2008. He has been working on robotic systems and software development for unmanned aerial and ground vehicles for ARL for the past nine years.
\end{IEEEbiography}
\vspace{-33pt}
\begin{IEEEbiography}[{\includegraphics[width=1in,height=1.25in,clip,keepaspectratio]{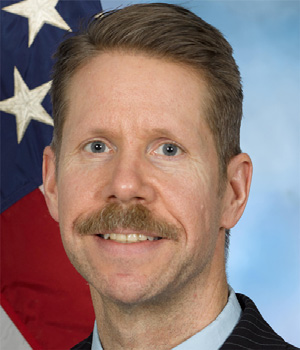}}]{Marshal Childers} is a Mechanical Engineer for the DEVCOM Army Research Laboratory. He leads efforts in the integration, prototyping, and experimentation of autonomous robotics technologies and has collaborated on numerous experiments that were conducted to evaluate technologies developed by ARL-funded research programs. Mr. Childers is an established engineer with over 20 years of experience in the areas of autonomous
technology assessment, integration, and mechanical design. He achieved a B.S. degree in mechanical engineering, graduating with honors, in 2000, and an M.S. degree in mechanical engineering in 2001, from the University of 
 Maryland, Baltimore County (UMBC).
\end{IEEEbiography}
\vspace{-33pt}


\end{document}